\documentclass[11pt]{article}
\pdfoutput=1 
\usepackage{palatino,microtype}
\usepackage[margin=15pt,font=small,labelfont={bf}]{caption}
\usepackage[margin=1in]{geometry}
\usepackage[english]{babel}
\usepackage[utf8]{inputenc}
\usepackage[compact]{titlesec}
\usepackage{cmap}
\usepackage[T1]{fontenc}
\usepackage{bm}
\usepackage{mathtools}
\usepackage{amsmath}
\usepackage{amssymb}
\usepackage{amsthm}
\usepackage{float}
\usepackage[numbers]{natbib}
\usepackage{hyperref}
\usepackage[svgnames,table]{xcolor}
\hypersetup{colorlinks={true},urlcolor={blue},linkcolor={DarkBlue},citecolor=[named]{DarkGreen},linktoc=all}
\usepackage{booktabs}

\newcommand{\BibTeX}{\rm B\kern-.05em{\sc i\kern-.025em b}\kern-.08em\TeX}

\usepackage[labelfont={normalfont,bf},textfont=it]{caption}
\usepackage{subcaption}

\usepackage{nicefrac}
\usepackage{pifont}

\usepackage{enumitem}
\usepackage{apxproof}
\usepackage{array,multirow,graphicx,bigdelim}
\usepackage{comment}
\usepackage{complexity}
\input{insbox}

\usepackage{tikz}
\usepackage{tikz-dimline}
\usepackage{tikz-cd}
\usetikzlibrary{patterns,fit,calc,math,shapes,shapes.multipart,decorations.text,arrows,decorations.markings,decorations.pathmorphing,shapes.geometric,positioning,decorations.pathreplacing}

\tikzset{snake it/.style={decorate, decoration=snake}}

\makeatletter
\let\oldnl\nl
\newcommand{\nonl}{\renewcommand{\nl}{\let\nl\oldnl}}
\makeatother

\usepackage{thmtools,thm-restate}
\usepackage[capitalise,nameinlink,noabbrev]{cleveref}

\newtheorem{theorem}{Theorem}
\newtheorem{corollary}{Corollary}

\newtheorem{claim}{Claim}
\Crefname{claim}{Claim}{Claims}
\Crefname{corollary}{Corollary}{Corollaries}
\Crefname{definition}{Definition}{Definitions}
\Crefname{example}{Example}{Examples}
\Crefname{lemma}{Lemma}{Lemmas}
\Crefname{property}{Property}{Properties}
\Crefname{proposition}{Proposition}{Propositions}
\Crefname{remark}{Remark}{Remarks}
\Crefname{theorem}{Theorem}{Theorems}

\usepackage{algorithm}
\usepackage[noend]{algpseudocode}
\usepackage{algorithmicx}

\theoremstyle{remark}
\newtheorem{remark}{Remark}

\title{Best of Both Worlds Guarantees for Equitable Allocations}

\author{
	\begin{tabular}{m{0.12\textwidth}m{0.12\textwidth}m{0.12\textwidth}m{0.12\textwidth}
 }
    \multicolumn{2}{c}{\textbf{Umang Bhaskar}} & \multicolumn{2}{c}{\textbf{Vishwa Prakash HV}}
        \\
        \multicolumn{2}{c}{Tata Institute of Fundamental Research} & \multicolumn{2}{c}{Chennai Mathematical Institute} 
        \\ 
		    \multicolumn{2}{c}{\href{mailto:umang@tifr.res.in}{\small{\texttt{umang@tifr.res.in}}}} &\multicolumn{2}{c}{\href{mailto:vishwa@cmi.ac.in}{\small{\texttt{vishwa@cmi.ac.in}}}}
        \\
        &&&\\
		\multicolumn{2}{c}{\textbf{Aditi Sethia}} & \multicolumn{2}{c}{\textbf{Rakshitha}} 
        \\
		\multicolumn{2}{c}{Indian Institute of Science, Bangalore} & \multicolumn{2}{c}{Indian Institute of Technology, Delhi} 
        \\ 
		\multicolumn{2}{c}{\href{mailto:aditisethia@iisc.ac.in}{\small{\texttt{aditisethia@iisc.ac.in}}}} & \multicolumn{2}{c}{\href{mailto:rakshitha.mt121@maths.iitd.ac.in}{\small{\texttt{	rakshitha.mt121@maths.iitd.ac.in}}}}
	\end{tabular}
}
\pgfplotsset{compat=1.18}

\newcommand{\AS}[1]{{\color{red}{AS: }{#1} \color{red}}}

\newcommand{\vishwa}[1]{{\color{blue}{Vishwa: }{#1} \color{blue}}}

\newcommand{\EF}[1]{\ifstrempty{#1}{\textrm{\textup{EF}}}{\textrm{\textup{EF{$#1$}}}}}

\newcommand{\EQ}[1]{\ifstrempty{#1}{\textrm{\textup{EQ}}}{\textrm{\textup{EQ{$#1$}}}}}
\newcommand{\EQX}{\textrm{\textup{EQX}}}
\newcommand{\eqset}{\mathcal{E}}
\newcommand{\mcT}{\mathcal{T}}

\newcommand{\ubar}[1]{\text{\b{$#1$}}}

\date{}

\begin{document}

\maketitle 
 

\begin{abstract}\label{sec:abstract}
Equitability is a well-studied fairness notion in fair division, where an allocation is equitable if all agents receive equal utility from their allocation. For indivisible items, an exactly equitable allocation may not exist, and a natural relaxation is EQ1, which stipulates that any inequitability should be resolved by the removal of a single item. In this paper, we study equitability in the context of randomized allocations. Specifically, we aim to achieve equitability in expectation (ex ante EQ) and require that each deterministic outcome in the support satisfies ex post EQ1. Such an allocation is commonly known as a `Best of Both Worlds' allocation, and has been studied, e.g., for envy-freeness and MMS.

We characterize the existence of such allocations using a geometric condition on linear combinations of EQ1 allocations, and use this to give comprehensive results on both existence and computation. For two agents, we show that ex ante EQ and ex post EQ1 allocations always exist and can be computed in polynomial time. For three or more agents, however, such allocations may not exist. We prove that deciding existence of such allocations is strongly NP-complete in general, and weakly NP-complete even for three agents. We also present a pseudo-polynomial time algorithm for a constant number of agents. We show that when agents have binary valuations, best of both worlds allocations that additionally satisfy welfare guarantees exist and are efficiently computable.

\end{abstract}
\section{Introduction}\label{sec:introduction}

Allocating a set of valuable resources among interested agents with diverse preferences is a fundamental problem, studied formally since at least the 1940s \cite{S48division, DS61cut}. 
Over the decades, it has attracted sustained interest from many disciplines, including economics, mathematics, and computer science \cite{M04fair,10.5555/3033138,brams1996fair,VARIAN197463}. A central goal in such settings is to ensure that the allocation satisfies certain notions of \emph{fairness} so that no individual is unduly disadvantaged by the outcome.
This problem arises in many real-world settings, such as 
resolving border disputes \cite{brams1996fair}, dividing rent \cite{Su1999RentalHS,10.1145/3131361}, assigning courses to students \cite{10.1287/opre.2016.1544}, allocating subsidized houses \cite{Kamiyama2021OnTC,10.5555/3545946.3599039}, and assigning conference papers to reviewers \cite{Lian_Mattei_Noble_Walsh_2018}. There have been recent implementations of fair division algorithms (see spliddit.org \cite{10.1145/2728732.2728738} and fairoutcomes.com). Given its ubiquity and importance, numerous fairness notions have been proposed to capture different ethical and practical desiderata in such divisions.

The most commonly studied notion of fairness is \emph{envy freeness} (EF) \cite{F67resource}, which requires that no agent prefers another's allocation over their own. Envy freeness avoids interpersonal comparisons of utility. In an envy free allocation, an agent could get very large utility, while another agent gets almost none.\footnote{Consider, e.g., an instance with $n$ agents, $n$ items, where the first agent gets utility 1 from the first item and 0 from all the others, and all other agents get utility $1/n$ from each item. Allocating item $i$ to agent $i$ is an envy free allocation where the first agent gets utility 1, while all other agents get utility $1/n$.} Another important notion in fairness is \emph{equitability} \cite{DS61cut} (EQ), which demands that all agents derive the same level of utility from their respective allocation. Equitability is particularly appealing in scenarios where equal perceived benefit is critical. In empirical studies with human subjects, equitability
has been demonstrated to have a significant impact on the perceived fairness of allocations \cite{HERREINER2010238,Herreiner2009}. Additionally, equitability plays a crucial role in applications such as divorce settlements \cite{brams1996fair} and rental harmony \cite{10.1145/3131361}. Equitable allocations are also well-studied in the literature, and previous work has studied existence~\cite{GMT14near,hosseini2025equitableallocationsmixturesgoods,barman2024nearly}, welfare guarantees~\cite{10.5555/3398761.3398810,FSV+19equitable,SCD23equitability,10.1007/978-3-031-43254-5_16}, as well as allocations that satisfy both approximate envy freeness and equitability~\cite{aziz21}. 

When dealing with indivisible resources, which can not be assigned fractionally, exact envy-freeness or equitability is often impossible. Thus several approximations have been introduced, such as envy-freeness up to one item (\EF{1}) \cite{B11combinatorial, LMM+04approximately} and equitability up to one item (\EQ{1}) \cite{GMT14near, FSV+19equitable}. These relaxations allow for a small degree of unfairness, permitting envy or inequitability to be eliminated by hypothetically removing at most one item from the larger-valued bundle. \EF{1} and \EQ{1} allocations always exist, and can also be computed efficiently for a large class of valuations \cite{LMM+04approximately, B11combinatorial,GMT14near}.

An alternative to approximation-based remedies for non-existence is to employ randomization, aiming to achieve fairness in expectation. Both envy-freeness and equitability can be trivially satisfied in expectation by allocating all goods to a single agent chosen uniformly at random. However the realized allocation is clearly unfair, since one agent receives everything, leaving all others with nothing.

Recent work \cite{Aziz2023BestOB} asked if randomization allowed us to get the best of both worlds. For envy-freeness, this meant a randomized allocation that is envy free ex ante (in expectation, prior to realisation of the random bits) and \EF{1} ex post (after the realization). 
Aziz et al.~\cite{Aziz2023BestOB} show that such an allocation always exists using the Probabilistic Serial algorithm \cite{BOGOMOLNAIA2001295} and Birkhoff's decomposition algorithm \cite{1573387450959988992,Neumann+1953+5+12} as subroutines. Hence, through randomization, stronger guarantees on fairness are obtainable. Subsequent works have further explored this in the context of envy-based \cite{10.1145/3670865.3673592,10.5555/3545946.3598686} and share-based fairness \cite{10.1007/978-3-031-22832-2_14}. Prior work also studies best of both worlds equitable allocations in restricted instances with chores~\cite{SunC25}.

Given the relevance of equitability as a fairness concept, a fundamental question to address is whether there is a randomized ex ante \EQ{} allocation which is ex post \EQ{1}. In this work, we comprehensively address this question and present a complete landscape of existence and tractability. We show that for two agents with normalised additive valuations, an allocation that is ex ante \EQ{} and ex post \EQ{1} always exists and can be computed efficiently. For binary valuations ($v_i(g) \in \{0, 1\}$ for any $i \in N$ and $g \in M$) and any number of agents, we exhibit existence and efficient computation coupled with strong welfare guarantees. However, beyond binary valuations, even for three agents, such allocations may not always exist, and the corresponding decision problem becomes NP-hard. Our results are in stark contrast with the results on envy freeness, where ex ante EF and ex post EF1 allocations always exist. Our techniques also differ from prior work, and include a geometric characterization of best of both worlds equitable allocations, the use of duality, and a bihierarchy theorem on the decomposition of a fractional LP solution into integral solutions~\cite{10.1257/aer.103.2.585}. 


We also introduce and study $i$-biased \EQ{1} allocations, which are \EQ{1} allocations where agent $i$ has maximum value among all agents. While \EQ{1} is a notion of fairness, an $i$-biased \EQ{1} allocation asks for a small amount of bias towards agent $i$. We show that for two agents, existence of $i$-biased \EQ{1} allocations characterizes existence of ex ante \EQ{} and ex post \EQ{1} allocations. Curiously, we show that in some instances, there may be an agent $i$ for which an $i$-biased \EQ{1} allocation does not exist, and hence every EQ{1} allocation disfavours $i$.

\subsection{Our Contributions}\label{subsec:contributions}

We present comprehensive results on both the existence and computation of both \EQ{} + \EQ{1} and \EQ{} + \EQX{} allocations. Our first result is a geometric characterization of instances that possess a BoBW allocation.

\begin{restatable}{theorem}{characterizationeqone}
\label{thm:any_convex_comb}
    Let \(\mathcal{I}\) be a fair division instance with $n$ agents and normalised valuations. The following statements are equivalent.
    \begin{enumerate}
        \item \(\mathcal{I}\) admits an ex ante \EQ{} and ex post \EQ{1} allocation.
        \item For any \(\lambda=(\lambda_1,\lambda_2,\ldots,\lambda_n)\in \mathbb{R}_n\) with $\sum_i \lambda_i = 0$, there exists an \EQ{1} allocation \(A\) such that \(\sum_{i=1}^{n}\lambda_i v_i(A_i)\ge 0\).
    \end{enumerate}
    
    Similarly, \(\mathcal{I}\) admits an ex ante \EQ{} and ex post \EQX{} allocation iff for any $\lambda \in \mathbb{R}_n$ with $\sum_i \lambda_i = 0$, there exists an \EQX{} allocation $A$ such that $\sum_{i=1}^n \lambda_i v_i(A_i) \ge 0$.
\end{restatable}

The characterization forms the basis for most of our results. 

\paragraph*{Two Agents.} Using the characterization, we show first that for instances with two agents and normalised valuations, an ex ante \EQ{} and ex post \EQ{1} allocation always exists. In fact, the characterization in this case is equivalent to proving the existence of an $i$-biased \EQ{1} allocation (i.e., where agent $i$ has the largest value), for $i \in \{1,2\}$. This proof is possibly the most technical result, and requires a careful analysis of the allocation obtained as we start with an \EQX{} allocation, transfer goods one-by-one from the rich agent to the poor agent, and then possibly swap the allocation of the two agents. 

An obvious question is if we can show the existence of an ex ante \EQ{} and ex post \EQX{} allocation. Particularly since in the prior case, we start with an \EQX{} allocation, this stronger result would not have been surprising. We show however that this is not true. Even with two agents, three items, and normalised valuations, an ex ante \EQ{} and ex post \EQX{} allocation may not exist.

\paragraph*{Binary Valuations.} If agents have binary valuations over the set of goods, it is not difficult to see that an ex ante \EQ{} and ex post \EQ{1} allocation must always exist. For this, any good that has value zero for some agent is simply assigned to the agent. This only leaves goods that have value one for every agent, which can then be assigned using a slight modification of the Birkhoff-von Neumann theorem (e.g.,~\cite{Aziz2023BestOB}). However, this allocation is wasteful, since many goods may be assigned to agents that have zero value for them.

Instead, we show a stronger result. We show that for agents with binary valuations, there exists an ex ante \EQ{} and ex post \EQ{1} allocation that has optimal social welfare (or total value over the agents) over all \EQ{} allocations. That is, we can get an ex ante \EQ{} and ex post \EQ{1} allocation with social welfare equal to that of highest welfare \EQ{} allocation (that could be fractional). Thus the restriction that the \EQ{} allocation be supported on \EQ{1} allocations does not cause any loss in the social welfare. We find this result surprising, since clearly not all fractional \EQ{} allocations can obtained as a convex combination of \EQ{1} allocations.

Our result in this case is based on carefully rounding a linear program using results on LPs with bihierarchical constraint structures --- a generalization of the Birkhoff-von Neumann theorem --- from Budish et al.~\cite{10.1257/aer.103.2.585}.

\paragraph*{General Instances: Existence.} For general instances, even for agents with normalised valuations, we show via an example that an ex ante \EQ{} and ex post \EQ{1} allocation may not exist. This is in contrast to envy-freeness, where an ex ante EF and ex post EF1 allocation always exists~\cite{Aziz2023BestOB}. Recall that even for two agents with normalised valuations and three items, an ex ante \EQ{} and ex post \EQX{} allocation may not exist.

In fact, our example showing the non-existence of ex ante \EQ{} and ex post \EQ{1} allocations is tight in multiple regards. The instance consists of just three agents and four items --- with two agents, or with three agents and three items, an ex ante \EQ{} and ex post \EQ{1} allocation always exists. The example consists of just two types of agents (where agents are of the same type if they have identical valuations). For just one type of agent, i.e., when all agents have identical valuations, there always exists an ex ante \EQ{} and ex post \EQ{1} allocation. Finally, there are just two types of goods as well. Note that instances with normalised valuations and a single type of good are trivial instances, where $v_i(g) = c$ for some constant $c$, for all agents $i$ and goods $g$.

\paragraph*{General Instances: Complexity.} We further study the computational complexity of determining the existence of ex ante \EQ{} and ex post \EQ{1} allocations. We show that, for three agents, it is weakly NP-hard to determine if there exists an ex ante \EQ{} and ex post \EQ{1} allocation, and with $n$ agents, the problem is strongly NP-hard.

We then show that the weak NP-hardness shown is the best possible, by giving a pseudopolynomial time algorithm for determining the existence of ex ante \EQ{} and ex post \EQ{1} allocations when the number of agents is constant. This is based on a somewhat technical dynamic program. In fact, the dynamic program is quite versatile, and can be slightly modified to determine the existence of (i) ex ante \EQ{} and ex post \EQX{} allocations, as well as existence of (ii) $i$-biased allocations, for any agent $i$. 

Finally, we show that for general instances, determining the existence of an $i$-biased allocation for an agent $i$ is NP-hard. 

\section{Preliminaries}\label{sec:prelims}

An instance \(\mathcal{I}\) of the fair division problem is specified by a tuple $\langle N, M, \mathcal{V} \rangle$, where $N$ is a set of $n$ \emph{agents}, $M$ is a set of $m$ indivisible \emph{items} (or \emph{goods}), and $\mathcal{V}$ is the \emph{valuation profile} consisting of each agent's valuation function $\{v_i\}_{i \in N}$. For any agent $i \in N$, its valuation function $v_i: M \rightarrow \mathbb{Z^+}$ specifies its numerical value (or utility) for each good in $M$. Valuations are additive, and hence for a subset $S \subseteq M$, $v_i(S) = \sum_{g \in S} v_i(g)$. If $v_i(M) = t$ for some constant $t \in \mathbb{Z^+}$ for all agents $i \in N$, then the instance is said to be \emph{normalised}. 

\vspace{0.2cm}

\noindent \textbf{Allocation.} A \emph{bundle} refers to any (possibly empty) subset of goods. An \emph{integral} allocation $A \coloneqq (A_1,\dots,A_n)$ is a partition of the set of goods $M$ into $n$ bundles, one for each agent, and $A_i$ denotes the bundle assigned to agent $i$. Note that $A$ is also specified by an $n \times m$ binary matrix (also denoted as $A$) with exactly one $1$ in every column. If a good can be fractionally assigned to multiple agents, such an allocation is called 
a \emph{fractional} allocation. It is specified by an $n \times m$ column-stochastic matrix $A$ where $A_{i, g}$ denotes the fraction of good $g$ assigned to agent $i$ and $\sum_{i \in N} A_{i, g} = 1$, i.e., each good is completely assigned. Note that the polytope of fractional allocations $\{A \in \mathbb{R}^{n \times m}_+ : \text{ for all $g \in M$}, \sum_{i \in N} A_{i,g} = 1\}$ is an integral polytope (e.g., the constraint matrix is totally unimodular, since each variable appears exactly once with coefficient $+1$). Hence, any fractional allocation can be obtained as a distribution over integral allocations. Further, by Carath\'eodory's theorem, any fractional allocation can be obtained as a distribution over at most $mn+1$ integral allocations. In the following, an allocation typically refers to an integral allocation.

A randomized allocation $X$ is a lottery over a set of integral allocations $\{A^k\}_{k \in [\ell]}$, where each of the allocation $A^k$ is chosen with probability $p_k \in [0, 1]$ and $\sum_{k \in [\ell]} p_k = 1$. Note that $X$ corresponds to the fractional allocation $\sum_{k \in [\ell]} p_k A^k$ in expectation. The integral allocations $\{A^k\}_{k \in [\ell]}$ are said to be the support of $X$.

\vspace{0.2cm}
\noindent \textbf{Equitable Allocation.} An allocation is said to be \emph{equitable} (\EQ{}) if for any pair of agents $i,j \in N$, we have $v_i(A_i) = v_j(A_j)$~\cite{DS61cut}, and \emph{equitable up to one good} (\EQ{1}) if for any pair of agents $i,j \in N$ such that $A_j \neq \emptyset$, there is a good $g \in A_j$ such that $v_i(A_i) \geqslant v_j(A_j \setminus \{g\})$~\cite{GMT14near,FSV+19equitable}. An allocation is \emph{equitable up to any good} (\EQX) (a stronger guarantee than \EQ{1}) if for any pair of agents $i,j \in N$ such that $A_j \neq \emptyset$, we have $v_i(A_i) \geqslant v_j(A_j \setminus \{g\})$ for all goods $g \in A_j$. We say that an agent $i$ is \emph{rich} under an \EQ{1} allocation $A$ if $v_i(A_i) \geq v_j(A_j) ~\forall~ j \in N \setminus \{i\}$. Analogously, an agent $j$ is \emph{poor} if $v_j(A_j) \leq v_i(A_i) ~\forall~i \in N \setminus \{j\}$. An \EQ{1} allocation $A$ is said to be $i$-biased \EQ{1} if agent $i$ is a rich agent in $A$.

\vspace{0.2cm}

Define $\eqset := \{A : A \text{ is an \EQ{1} allocation }\}$ be the set of all possible \EQ{1} allocations. For an allocation $A$ in an instance with $n$ agents, define $\vec{v}(A) := (v_1(A_1), \ldots, v_n(A_n)) \in \mathbb{R}^n$ as the vector of agent values. Let $\vec{v}(\eqset) := \{\vec{v}(A): A \in \eqset\}$. This is the set of all agent values obtainable in an \EQ{1} allocation.

\vspace{0.2cm}

\noindent \textbf{Best of Both Worlds Fairness.} A randomized allocation $X$ is \emph{ex ante \EQ} if every agent derives the same utility in expectation. That is, $\mathbb{E}[v_i(X_i)] = \mathbb{E}[v_j(X_j)] ~\forall~i, j \in N$. Equivalently, $X$ is ex ante \EQ{} if the fractional allocation corresponding to $X$ is EQ.\footnote{We use the fact that every fractional allocation can be obtained as a distribution over integral allocations.} 
The randomized allocation $X$ is \emph{ex post \EQ{1}} if it can be obtained as a distribution over \EQ{1} allocations. Similarly, the allocation $X$ is \emph{ex post \EQX{}} if it can be obtained as a distribution over \EQX{} allocations. We use \EQ{} + \EQ{1} to denote a randomized allocation $X$ that is ex ante \EQ{} and ex post \EQ{1}, and \EQ{} + \EQX{} for a randomized allocation $X$ that is ex ante \EQ{} and ex post \EQX{}. An allocation is \emph{Best of Both Worlds} (or BoBW) if it is either \EQ{} + \EQ{1}, or \EQ{} + \EQX{}.

\vspace{0.2cm} 

Note firstly that the assumption of normalization is crucial in the setting of indivisible items. Without this assumption, an \EQ{} + \EQ{1} allocation may not exist, even in trivial instances. E.g., consider an instance with two agents and two items. Agent 1 has value 10 for both items, and agent 2 has value 1 for both items. In any \EQ{1} allocation, agent 1 must get at least 1 item. But then in every \EQ{1} allocation, agent 1 has strictly greater value than agent 2. This must then be true of any distribution over \EQ{1} allocations as well, and hence any distribution over \EQ{1} allocations cannot 
give a fractional \EQ{} allocation.

We thus assume all valuations are normalised unless otherwise stated. We further note the following trivial cases: if either (i) all agents have identical valuations, or (ii) $m=n$, there always exists an \EQ{} + \EQ{1} allocation that can be computed efficiently.

\begin{restatable}{proposition}{trivial}
    Given a fair division instance with either (i) agents with identical valuations, or (ii) an equal number of goods and agents (i.e., $m=n$) and normalised valuations, an \EQ{} $+$ \EQ{1} allocation always exists and can be computed efficiently.
    \label{prop:trivial}
\end{restatable}

\begin{proof}
For agents with identical valuations, we observe that \EQ{} and EF allocations coincide, as do \EQ{1} and EF1 allocations. The first result then follows from Theorem 2 in~\cite{Aziz2023BestOB}, which shows the existence of EF + EF1 allocations. For the second result, with $m=n$, let $g_1$, $\ldots$, $g_n$ be the $n$ goods in the instance. For $k \in \{0, \ldots, n-1\}$, consider the $n$ integral allocations $A^k$, where $A_i^k = \{g_{((i+k-1)\text{ mod } n) + 1}\}$. Thus in $A^0$, agent $i$ receives good $g_i$, and in $A^k$, the allocation is shifted by $k$ items. Clearly, each $A^k$ is an \EQ{1} allocation. Further, consider the randomized allocation $X$ that picks $A^k$ with probability $1/n$. Each agent $i$ receives each good $g$ with equal probability in $X$, and hence this is an \EQ{} + \EQ{1} allocation. 
\end{proof}

\subsection*{A Characterization of BoBW Instances}

We now present a geometric characterization of instances that admit \EQ{} + \EQ{1} (or \EQ{} + \EQX{}) allocations in \Cref{thm:any_convex_comb}. This forms the basis for many of our further results.

\characterizationeqone*

\begin{proof} We show the proof for \EQ{} + \EQ{1} allocations. The proof for \EQ{} + \EQX{} allocations is very similar, with \EQX{} allocations taking the place of \EQ{1} allocations.

\vspace{0.2cm}
\noindent    \textbf{(1) $\implies$ (2):} Suppose \(\mathcal{I}\) admits an \EQ{} + \EQ{1} allocation. Let \(X\) be a randomized \EQ{} allocation over a support \((A^1, A^2, \ldots, A^\ell)\) of \EQ{1} allocations with corresponding probabilities \((p_1, p_2, \ldots, p_{\ell})\). By the definition of ex ante \EQ{}, we have \( \mathbb{E}[v_i(X_i)] = \mathbb{E}[v_j(X_j)] \quad \forall i,j \in N.\)
    That is, all agents receive the same expected utility under the randomized allocation \(X\).

    Suppose for the sake of contradiction there exists some \(\lambda \in \mathbb{R}_n\) with $\sum_i \lambda_i = 0$ such that for all \EQ{1} allocations \(A \in \eqset \), $\sum_{i} \lambda_i v_i(A_i) < 0$. Then for the randomized allocation $X$, 

    \begin{align*}
        \sum_i \lambda_i \mathbb{E}[v_i(X_i)] 
        &= \sum_i \lambda_i  \left( \sum_{k=1}^\ell p_k v_i(A^k_i) \right) \\
        &= \sum_{k=1}^\ell p_k \left( \sum_i \lambda_i v_i(A^k_i) \right) < 0 \, .
    \end{align*}

    Since \(X\) is ex ante \EQ{}, each agent has the same expected utility, say \(u\), and hence
    \[
        \sum_i \lambda_i \mathbb{E}[v_i(X_i)] = u \, \sum_i \lambda_i =  0.
    \]
    This gives a contradiction, and hence for every \(\lambda \in \mathbb{R}_n\), with $\sum_i \lambda_i = 0$, there exists some \EQ{1} allocation $A \in \eqset$ such that \( \sum_{i=1}^{n} \lambda_i  v_i(A_i) \geq 0. \)

\vspace{0.2cm}
    \noindent \textbf{(1) \(\impliedby\) (2):}  Let \( \eqset = \{A^1, A^2, \dots, A^\ell\} \) be the (finite) set of all \EQ{1} allocations. Consider the following linear program, with variables $\mu$ and $\{p_k\}_{k \in [\ell]}$: 

    \[
        \begin{aligned}
        \text{maximize} \quad & 0 \\
        \text{subject to} \quad  &-\sum_{k=1}^\ell p_k \cdot v_i(A^k_i) +\mu = 0 \quad \text{for all } i \in [n], \\
        &\quad \sum_{k=1}^\ell p_k = 1, \\
        &\quad  p_k \geq 0 \quad \text{for all } k \in [\ell]
        \end{aligned} \]

    This linear program seeks a convex combination of \EQ{1} allocations such that all agents receive the same expected utility \( \mu \), i.e., the convex combination is ex ante \EQ{}. Any feasible solution to this LP gives a randomized allocation that is \EQ{} + \EQ{1} (since all allocations in the support are \EQ{1}). We show that this LP is feasible, completing the proof.

    To show feasibility, we consider the dual linear program. For the dual, we introduce dual variables \( \lambda_i \in \mathbb{R} \) for each agent \(i\), and \( \theta \in \mathbb{R} \) for the normalization constraint. 
    \[
    \begin{aligned}
    \text{minimize} \quad & \theta \\
    \text{subject to} \quad & -\sum_{i=1}^n \lambda_i \cdot v_i(A^k_i) + \theta \geq 0 \quad \text{for all } k \in [\ell], \\
    &\quad \sum_{i=1}^n \lambda_i = 0
    \end{aligned} \]

    It is evident that the dual has a feasible solution given by \(\lambda_i = 0\) for all \(i \in [n]\) and \(\theta = 0\). From the dual constraints, we have:
    \[
    \theta \geq \sum_{i=1}^n \lambda_i \cdot v_i(A^k_i) \quad \text{for all } k \in [\ell]
    \quad \implies \quad
    \theta \geq \max_{k \in [\ell]} \sum_{i=1}^n \lambda_i \cdot v_i(A^k_i)
    \]

    By assumption, if $\theta, (\lambda_i)_{i \in [n]}$ is a feasible dual solution, then there exists \( k \in [\ell] \) such that for the \EQ{1} allocation $A^k$, \( \sum_i \lambda_i v_i(A^k_i) \geq 0, \)
    and hence \( \theta \geq 0 \). Therefore, the dual is bounded below. By strong duality, the primal is feasible. This completes the proof.
    \end{proof}

\begin{remark} Note that the primal has an exponential number of variables owing to the potentially exponential number of \EQ{1} allocations. Since each (possibly fractional) allocation $A$ is represented by the valuation vector $\vec{v}(A) \in \mathbb{R}^n$, if $X$ is an ex ante \EQ{} and ex post \EQ{1} (or \EQX{}) allocation, then by Carath\'eodory's theorem, it can be obtained as a distribution over at most $n+1$ \EQ{1} (or \EQX{}) allocations. In particular, there is a polynomial-sized certificate for existence of \EQ{} + \EQ{1} or \EQ{} + \EQX{} allocations in an instance, which is this succinct distribution.
\end{remark}

When there are only two agents, we have the following corollary.

\begin{corollary}\label{cor:convex_comb}
Let \(\mathcal{I}\) be a fair division instance with \(n = 2\) agents, \(m\) items, and normalised valuations. The following statements are equivalent:
\begin{enumerate}
    \item \(\mathcal{I}\) admits an \EQ{} $+$ \EQ{1} allocation.
    \item There exists an \(i\)-biased \EQ{1} allocation for each \(i \in \{1, 2\}\).
\end{enumerate}

Similarly, \(\mathcal{I}\) admits an \EQ{} $+$ \EQX{} allocation iff there exists an \(i\)-biased \EQX{} allocation for each \(i \in \{1, 2\}\).
\end{corollary}

To see the corollary, observe that for two agents, the second condition in Theorem~\ref{thm:any_convex_comb} is equivalent to the statement that for any $\lambda \in \mathbb{R}$, there exists an \EQ1{} (or \EQX{}) allocation $A$ such that $\lambda v_1(A_1) - \lambda v_2(A_2) \ge 0$. Then for $\lambda > 0$, this is equivalent to the condition that there exists a 1-biased \EQ{1} (or \EQX{}) allocation. For $\lambda <0$, this is equivalent to the condition that there exists a 2-biased \EQ{1} (or \EQX{}) allocation.

\section{Two Agents}\label{sec:twoagents}


In this section, we consider instances with two agents. We show that an allocation that is \EQ{} + \EQ{1} always exists and can be computed in linear time. However, this result does not extend to the stronger notion of \EQ{} + \EQX{}; such allocations do not always exist.

\begin{restatable}{theorem}{bobwtwo}
\label{thm:bobw_n_2}
    Given a fair division instance with two agents with normalised valuations, an \EQ{} $+$ \EQ{1} allocation always exists and can be computed in time \(\mathcal{O}(m)\). 
\end{restatable}

\begin{proof}
    By Corollary~\ref{cor:convex_comb}, it suffices to show that an \(i\)-biased \EQ{1} allocation always exists. Therefore, without loss of generality, we consider agent \(1\) and show that Algorithm~\ref{alg:1biased-eq1} always computes a \(1\)-biased \EQ{1} allocation in \(\mathcal{O}(m)\) time.

    By~\cite{GMT14near}, an \EQX{}  allocation $A = (A_1, A_2)$ always exists and can be computed in linear time. If \(A\) is already a \(1\)-biased allocation (i.e., $v_1(A_1) \ge v_2(A_2)$), the algorithm returns allocation $A$ in Line~\ref{line:returneqx}. Otherwise, \(v_1(A_1) < v_2(A_2)\). Define \(\delta = v_2(A_2) - v_1(A_1)\). Since \(A\) is \EQX{}, it follows that for all \(g \in A_2\), we have \(v_2(g) \ge \delta\).
    
    Moreover, by normalization, for any partition \((B_1, B_2)\) of \(M\), we have $v_1(B_1) + v_1(B_2) = v_2(B_1) + v_2(B_2)$, and hence 
    \begin{align}
 v_1(B_1) - v_2(B_2) &= v_2(B_1) - v_1(B_2) \, .\label{eq:balance}
\end{align}

That is, swapping the two bundles in any allocation preserves the magnitude of the utility difference between the agents but reverses its sign. However, such a swap need not necessarily preserve the \EQ{1} property.

We now proceed by considering two cases based on agent \(1\)'s valuation for goods in \(A_2\):

    \vspace{0.2cm}
\noindent \textbf{Case 1:} There exists a good \(\hat{g} \in A_2\) such that \(v_1(\hat{g}) \ge \delta\).

This case (see~\Cref{fig:fig1}) is handled in Step~\ref{step:s1} of the algorithm. Consider the allocation \(A'\) obtained by swapping the bundles: \(A'_1 = A_2\) and \(A'_2 = A_1\).  
By construction,  
\[
v_1(A'_1) - v_2(A'_2) = v_1(A_2) - v_2(A_1) = \delta,
\]  
so agent 1 is the richer agent in \(A'\).  
Moreover, since \(\hat{g} \in A'_1\) and \(v_1(\hat{g}) \ge \delta\), the allocation \(A'\) satisfies \EQ{1}. Thus in Line~\ref{line:returnswap}, the algorithm returns a \(1\)-biased \EQ{1} allocation. 

    \vspace{0.2cm}
\noindent\textbf{Case 2:} For all goods \(g \in A_2\), we have \(v_1(g) < \delta\).

Define the set of \emph{compressing goods} as \(
    C = \{g \in M \mid v_1(g) \ge v_2(g)\}, \)
and let the set of \emph{expanding goods} be \(E = M \setminus C\).  

From the case assumption, \(v_1(g) < \delta\) for all \(g \in A_2\), and since \(A\) is an \EQX{} allocation, we also have \(v_2(g) \ge \delta\) for all \(g \in A_2\). Therefore, \(v_1(g) < v_2(g)\) for all \(g \in A_2\), which implies \(C \subseteq A_1\).

Now, let \(\hat{g} \in A_2\) be an arbitrary good in \(A_2\). We first show that \(v_1(C)+v_2(C) \ge v_2(\hat{g})-\delta\) (see Figure~\ref{fig:fig2}). Since the valuations are normalised, and \(M=C\ \uplus\ E\),  we have, from~\eqref{eq:balance}, that $v_1(C) - v_2(C) = v_2(E) - v_1(E) = \sum_{g \in E} \left(v_1(g) - v_2(g)\right)$. Now recall that $\hat{g} \in E$, and for every good $g \in E$, $v_1(g) - v_2(g) > 0$. Hence, 

\begin{align*}
v_1(C) - v_2(C) &\ge v_2(\hat{g}) - v_1(\hat{g}) > v_2(\hat{g}) - \delta \\
\end{align*}

\noindent since $v_1(\hat{g}) < \delta$ by assumption in this case. Therefore, we have:

\begin{equation}
    \label{eq:eq2}
    v_1(C)+v_2(C) \ge v_1(C)-v_2(C) >  v_2(\hat{g})-\delta \, .
\end{equation}

Now consider the allocation \(A' = (A'_1, A'_2)\), where  
$ A'_1 = A_1 \cup \{\hat{g}\}$, and $A'_2 = A_2 \setminus \{\hat{g}\}$. That is, transfer \(\hat{g}\) from agent 2 to agent 1 (see Figure~\ref{fig:fig2}).
Next, following Steps~\ref{step:begin_for}–\ref{step:end_for} of Algorithm~\ref{alg:1biased-eq1}, we iteratively transfer goods from \(C\) (i.e., from \(A'_1\)) to \(A'_2\) as long as agent 1 remains the richer agent.

    \begin{figure}
    \centering
    \begin{subfigure}[t]{0.48\textwidth}
        \centering
 
\tikzset{
pattern size/.store in=\mcSize, 
pattern size = 5pt,
pattern thickness/.store in=\mcThickness, 
pattern thickness = 0.3pt,
pattern radius/.store in=\mcRadius, 
pattern radius = 1pt}
\makeatletter
\pgfutil@ifundefined{pgf@pattern@name@_ydkoqvyjt}{
\pgfdeclarepatternformonly[\mcThickness,\mcSize]{_ydkoqvyjt}
{\pgfqpoint{0pt}{-\mcThickness}}
{\pgfpoint{\mcSize}{\mcSize}}
{\pgfpoint{\mcSize}{\mcSize}}
{
\pgfsetcolor{\tikz@pattern@color}
\pgfsetlinewidth{\mcThickness}
\pgfpathmoveto{\pgfqpoint{0pt}{\mcSize}}
\pgfpathlineto{\pgfpoint{\mcSize+\mcThickness}{-\mcThickness}}
\pgfusepath{stroke}
}}
\makeatother
\tikzset{every picture/.style={line width=0.75pt}} 

\begin{tikzpicture}[x=0.75pt,y=0.75pt,yscale=-1,xscale=1]

\draw   (121,130) -- (161,130) -- (161,250) -- (121,250) -- cycle ;
\draw   (191,80) -- (231,80) -- (231,250) -- (191,250) -- cycle ;
\draw  [dash pattern={on 4.5pt off 4.5pt}]  (111,80) -- (231,80) ;
\draw    (111,80) -- (111,130) ;
\draw [shift={(111,130)}, rotate = 270] [color={rgb, 255:red, 0; green, 0; blue, 0 }  ][line width=0.75]    (0,5.59) -- (0,-5.59)(10.93,-3.29) .. controls (6.95,-1.4) and (3.31,-0.3) .. (0,0) .. controls (3.31,0.3) and (6.95,1.4) .. (10.93,3.29)   ;
\draw [shift={(111,80)}, rotate = 90] [color={rgb, 255:red, 0; green, 0; blue, 0 }  ][line width=0.75]    (0,5.59) -- (0,-5.59)(10.93,-3.29) .. controls (6.95,-1.4) and (3.31,-0.3) .. (0,0) .. controls (3.31,0.3) and (6.95,1.4) .. (10.93,3.29)   ;
\draw    (90,130) -- (90,250) ;
\draw [shift={(90,250)}, rotate = 270] [color={rgb, 255:red, 0; green, 0; blue, 0 }  ][line width=0.75]    (0,5.59) -- (0,-5.59)(10.93,-3.29) .. controls (6.95,-1.4) and (3.31,-0.3) .. (0,0) .. controls (3.31,0.3) and (6.95,1.4) .. (10.93,3.29)   ;
\draw [shift={(90,130)}, rotate = 90] [color={rgb, 255:red, 0; green, 0; blue, 0 }  ][line width=0.75]    (0,5.59) -- (0,-5.59)(10.93,-3.29) .. controls (6.95,-1.4) and (3.31,-0.3) .. (0,0) .. controls (3.31,0.3) and (6.95,1.4) .. (10.93,3.29)   ;
\draw  [pattern=_ydkoqvyjt,pattern size=6pt,pattern thickness=0.75pt,pattern radius=0pt, pattern color={rgb, 255:red, 0; green, 0; blue, 0}] (121,180) -- (161,180) -- (161,250) -- (121,250) -- cycle ;
\draw  [fill={rgb, 255:red, 206; green, 206; blue, 206 }  ,fill opacity=1 ] (191,80) -- (231,80) -- (231,150) -- (191,150) -- cycle ;
\draw  [dash pattern={on 4.5pt off 4.5pt}]  (120,130) -- (240,130) ;
\draw  [dash pattern={on 4.5pt off 4.5pt}]  (191,150) -- (240,150) ;
\draw    (250,100) -- (250,130) ;
\draw [shift={(250,130)}, rotate = 270] [color={rgb, 255:red, 0; green, 0; blue, 0 }  ][line width=0.75]    (0,5.59) -- (0,-5.59)(10.93,-3.29) .. controls (6.95,-1.4) and (3.31,-0.3) .. (0,0) .. controls (3.31,0.3) and (6.95,1.4) .. (10.93,3.29)   ;
\draw    (250,150) -- (250,180) ;
\draw [shift={(250,150)}, rotate = 90] [color={rgb, 255:red, 0; green, 0; blue, 0 }  ][line width=0.75]    (0,5.59) -- (0,-5.59)(10.93,-3.29) .. controls (6.95,-1.4) and (3.31,-0.3) .. (0,0) .. controls (3.31,0.3) and (6.95,1.4) .. (10.93,3.29)   ;

\draw (132,262.4) node [anchor=north west][inner sep=0.75pt]    {$A_{1}$};
\draw (202,262.4) node [anchor=north west][inner sep=0.75pt]    {$A_{2}$};
\draw (98,92.4) node [anchor=north west][inner sep=0.75pt]    {$\delta $};
\draw (49,192.4) node [anchor=north west][inner sep=0.75pt]  [font=\footnotesize]  {$v_{1}( A_{1})$};
\draw (103,212.4) node [anchor=north west][inner sep=0.75pt]    {$C$};
\draw (207,100.4) node [anchor=north west][inner sep=0.75pt]    {$\hat{g}$};
\draw (240,130.4) node [anchor=north west][inner sep=0.75pt]  [font=\footnotesize]  {$v_{2}( \ \hat{g} \ ) -\delta $};

\end{tikzpicture}
        \label{fig:fig2a}
    \end{subfigure}
    \hfill
    \begin{subfigure}[t]{0.48\textwidth}
        \centering
 
\tikzset{
pattern size/.store in=\mcSize, 
pattern size = 5pt,
pattern thickness/.store in=\mcThickness, 
pattern thickness = 0.3pt,
pattern radius/.store in=\mcRadius, 
pattern radius = 1pt}
\makeatletter
\pgfutil@ifundefined{pgf@pattern@name@_p4sf3bj7a}{
\pgfdeclarepatternformonly[\mcThickness,\mcSize]{_p4sf3bj7a}
{\pgfqpoint{0pt}{-\mcThickness}}
{\pgfpoint{\mcSize}{\mcSize}}
{\pgfpoint{\mcSize}{\mcSize}}
{
\pgfsetcolor{\tikz@pattern@color}
\pgfsetlinewidth{\mcThickness}
\pgfpathmoveto{\pgfqpoint{0pt}{\mcSize}}
\pgfpathlineto{\pgfpoint{\mcSize+\mcThickness}{-\mcThickness}}
\pgfusepath{stroke}
}}
\makeatother
\tikzset{every picture/.style={line width=0.75pt}} 

\begin{tikzpicture}[x=0.75pt,y=0.75pt,yscale=-1,xscale=1]

\draw   (141,131) -- (181,131) -- (181,251) -- (141,251) -- cycle ;
\draw   (211,150) -- (251,150) -- (251,251) -- (211,251) -- cycle ;
\draw  [pattern=_p4sf3bj7a,pattern size=6pt,pattern thickness=0.75pt,pattern radius=0pt, pattern color={rgb, 255:red, 0; green, 0; blue, 0}] (141,181) -- (181,181) -- (181,251) -- (141,251) -- cycle ;
\draw  [fill={rgb, 255:red, 206; green, 206; blue, 206 }  ,fill opacity=1 ] (141,61) -- (181,61) -- (181,131) -- (141,131) -- cycle ;
\draw    (272,100) -- (272,130) ;
\draw [shift={(272,130)}, rotate = 270] [color={rgb, 255:red, 0; green, 0; blue, 0 }  ][line width=0.75]    (0,5.59) -- (0,-5.59)(10.93,-3.29) .. controls (6.95,-1.4) and (3.31,-0.3) .. (0,0) .. controls (3.31,0.3) and (6.95,1.4) .. (10.93,3.29)   ;
\draw    (272,150) -- (272,180) ;
\draw [shift={(272,150)}, rotate = 90] [color={rgb, 255:red, 0; green, 0; blue, 0 }  ][line width=0.75]    (0,5.59) -- (0,-5.59)(10.93,-3.29) .. controls (6.95,-1.4) and (3.31,-0.3) .. (0,0) .. controls (3.31,0.3) and (6.95,1.4) .. (10.93,3.29)   ;
\draw    (120,60) -- (120,90) ;
\draw [shift={(120,60)}, rotate = 90] [color={rgb, 255:red, 0; green, 0; blue, 0 }  ][line width=0.75]    (0,5.59) -- (0,-5.59)(10.93,-3.29) .. controls (6.95,-1.4) and (3.31,-0.3) .. (0,0) .. controls (3.31,0.3) and (6.95,1.4) .. (10.93,3.29)   ;
\draw    (120,110) -- (120,130) ;
\draw [shift={(120,130)}, rotate = 270] [color={rgb, 255:red, 0; green, 0; blue, 0 }  ][line width=0.75]    (0,5.59) -- (0,-5.59)(10.93,-3.29) .. controls (6.95,-1.4) and (3.31,-0.3) .. (0,0) .. controls (3.31,0.3) and (6.95,1.4) .. (10.93,3.29)   ;
\draw  [dash pattern={on 4.5pt off 4.5pt}]  (141,61) -- (320,60) ;
\draw  [dash pattern={on 4.5pt off 4.5pt}]  (141,131) -- (280,130) ;
\draw  [dash pattern={on 4.5pt off 4.5pt}]  (211,150) -- (330,150) ;
\draw    (320,60) -- (320,90) ;
\draw [shift={(320,60)}, rotate = 90] [color={rgb, 255:red, 0; green, 0; blue, 0 }  ][line width=0.75]    (0,5.59) -- (0,-5.59)(10.93,-3.29) .. controls (6.95,-1.4) and (3.31,-0.3) .. (0,0) .. controls (3.31,0.3) and (6.95,1.4) .. (10.93,3.29)   ;
\draw    (320,120) -- (320,150) ;
\draw [shift={(320,150)}, rotate = 270] [color={rgb, 255:red, 0; green, 0; blue, 0 }  ][line width=0.75]    (0,5.59) -- (0,-5.59)(10.93,-3.29) .. controls (6.95,-1.4) and (3.31,-0.3) .. (0,0) .. controls (3.31,0.3) and (6.95,1.4) .. (10.93,3.29)   ;

\draw (152,263.4) node [anchor=north west][inner sep=0.75pt]    {$A'_{1}$};
\draw (222,263.4) node [anchor=north west][inner sep=0.75pt]    {$A'_{2}$};
\draw (123,213.4) node [anchor=north west][inner sep=0.75pt]    {$C$};
\draw (102,92.4) node [anchor=north west][inner sep=0.75pt]  [font=\footnotesize]  {$v_{1}( \ \hat{g} \ )$};
\draw (251,132.4) node [anchor=north west][inner sep=0.75pt]  [font=\scriptsize]  {$v_{2}( \ \hat{g} \ ) -\delta $};
\draw (302,93.4) node [anchor=north west][inner sep=0.75pt]  [font=\footnotesize]  {$v( \ \hat{g} \ ) -\delta $};

\end{tikzpicture}
        \label{fig:fig2b}
    \end{subfigure}
    \caption{Figure showing the agent values before and after transferring good $\hat{g}$. The compressing goods $C \subseteq A_1$, and $v_1(C) + v_2(C) \ge  v_2(\hat{g})-\delta$.}
    \label{fig:fig2}
\end{figure}
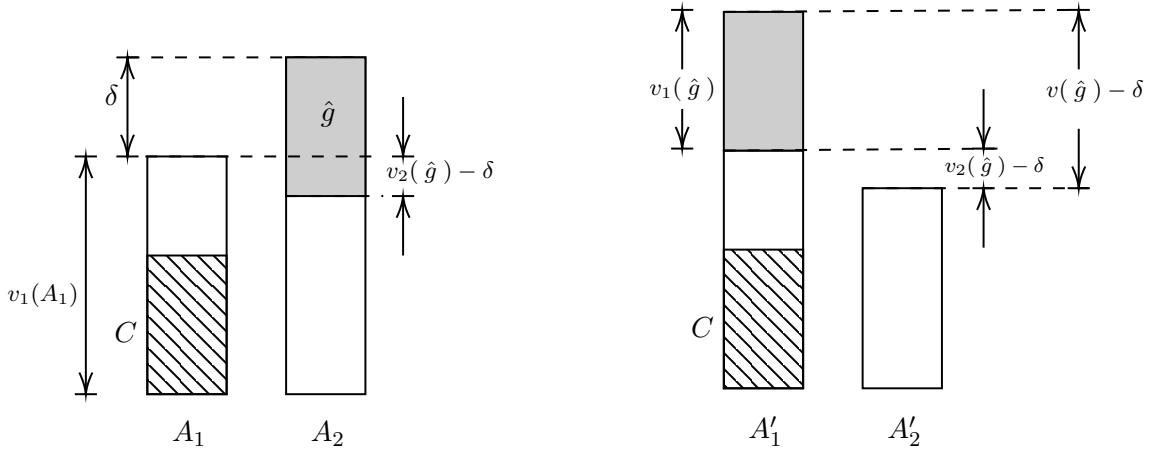

\begin{figure}[h]
    \centering
    \tikzset{every picture/.style={line width=0.75pt}} 

\begin{tikzpicture}[x=0.75pt,y=0.75pt,yscale=-1,xscale=1]

\draw   (90,110) -- (130,110) -- (130,230) -- (90,230) -- cycle ;
\draw   (160,60) -- (200,60) -- (200,230) -- (160,230) -- cycle ;
\draw  [dash pattern={on 4.5pt off 4.5pt}]  (80,60) -- (200,60) ;
\draw    (80,60) -- (80,110) ;
\draw [shift={(80,110)}, rotate = 270] [color={rgb, 255:red, 0; green, 0; blue, 0 }  ][line width=0.75]    (0,5.59) -- (0,-5.59)(10.93,-3.29) .. controls (6.95,-1.4) and (3.31,-0.3) .. (0,0) .. controls (3.31,0.3) and (6.95,1.4) .. (10.93,3.29)   ;
\draw [shift={(80,60)}, rotate = 90] [color={rgb, 255:red, 0; green, 0; blue, 0 }  ][line width=0.75]    (0,5.59) -- (0,-5.59)(10.93,-3.29) .. controls (6.95,-1.4) and (3.31,-0.3) .. (0,0) .. controls (3.31,0.3) and (6.95,1.4) .. (10.93,3.29)   ;
\draw    (220,60) -- (220,230) ;
\draw [shift={(220,230)}, rotate = 270] [color={rgb, 255:red, 0; green, 0; blue, 0 }  ][line width=0.75]    (0,5.59) -- (0,-5.59)(10.93,-3.29) .. controls (6.95,-1.4) and (3.31,-0.3) .. (0,0) .. controls (3.31,0.3) and (6.95,1.4) .. (10.93,3.29)   ;
\draw [shift={(220,60)}, rotate = 90] [color={rgb, 255:red, 0; green, 0; blue, 0 }  ][line width=0.75]    (0,5.59) -- (0,-5.59)(10.93,-3.29) .. controls (6.95,-1.4) and (3.31,-0.3) .. (0,0) .. controls (3.31,0.3) and (6.95,1.4) .. (10.93,3.29)   ;
\draw    (60,110) -- (60,230) ;
\draw [shift={(60,230)}, rotate = 270] [color={rgb, 255:red, 0; green, 0; blue, 0 }  ][line width=0.75]    (0,5.59) -- (0,-5.59)(10.93,-3.29) .. controls (6.95,-1.4) and (3.31,-0.3) .. (0,0) .. controls (3.31,0.3) and (6.95,1.4) .. (10.93,3.29)   ;
\draw [shift={(60,110)}, rotate = 90] [color={rgb, 255:red, 0; green, 0; blue, 0 }  ][line width=0.75]    (0,5.59) -- (0,-5.59)(10.93,-3.29) .. controls (6.95,-1.4) and (3.31,-0.3) .. (0,0) .. controls (3.31,0.3) and (6.95,1.4) .. (10.93,3.29)   ;
\draw  [fill={rgb, 255:red, 206; green, 206; blue, 206 }  ,fill opacity=1 ] (160,150) -- (200,150) -- (200,190) -- (160,190) -- cycle ;

\draw (101,242.4) node [anchor=north west][inner sep=0.75pt]    {$A_{1}$};
\draw (171,242.4) node [anchor=north west][inner sep=0.75pt]    {$A_{2}$};
\draw (67,72.4) node [anchor=north west][inner sep=0.75pt]    {$\delta $};
\draw (229,132.4) node [anchor=north west][inner sep=0.75pt]  [font=\footnotesize]  {$v_{2}( A_{2})$};
\draw (19,172.4) node [anchor=north west][inner sep=0.75pt]  [font=\footnotesize]  {$v_{1}( A_{1})$};
\draw (177,160.4) node [anchor=north west][inner sep=0.75pt]    {$\hat{g}$};

\end{tikzpicture}
    \caption{\centering An \EQX{} allocation $A$ where $v_2(A_2) - v_1(A_1) = \delta$.}
    \label{fig:fig1}
\end{figure}
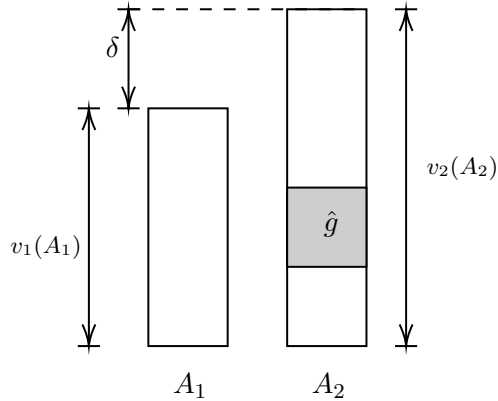

We now consider two subcases based on whether the transfer of the entire set \(C\) keeps agent \(1\) as the richer agent, or if at some point during the transfer of goods from \(C\), the transfer of a single good \(s^*\) causes agent \(1\) to no longer be the richer agent.

\vspace{0.2cm}
\noindent \textbf{Case 2.1:} The entire set \(C\) is transferred to \(A'_2\), and agent \(1\) remains the richer agent. In this case, we claim that the resulting allocation $A' = (A_1', A_2')$ is an \EQ{1} allocation (and hence, $A'$ is a 1-biased \EQ{1} allocation).


Note that $A_1' = A_1 \cup \hat{g} \setminus C$, and $A_2' = A_2 \cup C \setminus \hat{g}$. Then


\begin{align*}
v_1(A_1' \setminus \{\hat{g}\}) &= v_1(A_1) - v_1(C) && \text{(since $A_1' = A_1 \cup \hat{g} \setminus C$)}\\
&= v_2(A_2) - \delta -v_1(C) &&\text{(since \(v_1(A_1 = v_2(A_2)-\delta\))} \\
& \le v_2(A_2) - (v_2(\hat{g})-v_2(C)) &&\text{(from~\eqref{eq:eq2})}\\
&=v_2(\hat{A}_2) && \text{(since $A_2' = A_2 \cup C \setminus \hat{g}$).}
\end{align*}

Therefore, if all goods in $C$ are transferred from $A_1$ to $A_2$, the resulting allocation \(A'\) is a \(1\)-biased \EQ{1} allocation, and is returned in Line~\ref{line:alltransferred}. 

\vspace{0.2cm}
\noindent \textbf{Case 2.2:} There exists a good \(s^* \in C\) such that transferring \(s^*\) from \(A'_1\) to \(A'_2\) causes agent 1 to cease being the richer agent. That is, if \(A'_1\) and \(A'_2\) are the bundles just before transferring \(s^*\), then $v_1(A_1') \ge v_2(A_2')$, and $v_1(A_1' \setminus \{s^*\} < v_2(A_2' \cup \{s^*\})$. 

Clearly, if \(v_1(A'_1\setminus{s^*}) \le v_2(A'_2)\), then \((A'_1,A'_2)\) is the required 1-biased \EQ{1} allocation, and this is the allocation returned in Line~\ref{line:alltransferred}. Thus, consider the case when 
\begin{equation}
    \label{eq:assum}
    v_1(A'_1\setminus{s^*})> v_2(A'_2) \, .
\end{equation}

Then, let \(A_1''=A'_2\cup s^*\) and \(A_2''=A'_1\setminus \{s^*\}\). That is, transfer \(s^*\) and then swap the bundles. We claim that this is the required \(1\)-biased \EQ{1} allocation. 

Firstly, to show that \(1\) is the richer agent:
\begin{align*}
v_1(A_1'')-v_2(A_2'') &=v_2(A_1'')-v_1(A_2'') &&\text{(from~\eqref{eq:balance})}\\
&=v_2(A'_2\cup s^*) - v_1(A'_1\setminus \{s^*\}) > 0 &&\text{(from~\eqref{eq:assum}).}
\end{align*}

Now, consider the removal of \(s^*\) from \(A_1''\):

\begin{align*}
v_1(A_1''\setminus s^*)-v_2(A_2'')
&=v_1(A_2')-v_2({A'}_1\setminus s^*) \\
&=v_1(A'_2\cup s^*) -v_2(A'_1\setminus s^*) -  v_1(s^*) \\
&= v_2(A'_2\cup s^*) -v_1(A'_1\setminus s^*) -v_1(s^*) &&\text{(by~\eqref{eq:balance})}\\
&= v_2(A'_2) -v_1(A'_1\setminus s^*) + v_2(s^*)-v_1(s^*)\\
&\le v_2(A'_2) -v_1(A'_1\setminus s^*) &&\text{(since \(s^*\in C\))}\\
&\le 0 &&\text{(from~\eqref{eq:assum}).}
\end{align*}

Therefore, \(A''\) is the required \(1\)-biased \EQ{1} allocation, which is returned in Line~\ref{line:s-star-end}. Furthermore, since each good is transferred at most once, this procedure terminates in \(\mathcal{O}(m)\) time.
\end{proof}

\begin{algorithm}[ht]
    \caption{\textsc{Obtain a \(1\)-biased \EQ{1} Allocation}}
    \label{alg:1biased-eq1}
    
    \hspace*{\algorithmicindent} \textbf{Input}: A normalised fair division instance \(\mathcal{I} = \langle \{1,2\}, M, \mathcal{V} \rangle\).\\
    \hspace*{\algorithmicindent} \textbf{Output}: A \(1\)-biased \EQ{1} allocation.
    \begin{algorithmic}[1]
    \State Compute an \(\EQX{}\) allocation \(A = (A_1, A_2)\) \Comment{Can be done in \(\mathcal{O}(m)\) time~\cite{GMT14near}}
    \If{\(v_1(A_1) \ge v_2(A_2)\)}
        \State \Return \((A_1, A_2)\) \label{line:returneqx}
    \EndIf
    
    \State Let \(\delta \gets v_2(A_2) - v_1(A_1)\)
    
    \If{there exists \(\hat{g} \in A_2\) such that \(v_1(\hat{g}) \ge \delta\)}\label{step:s1}
        \State \Return \((A_2, A_1)\) \Comment{Swap the bundles} \label{line:returnswap}
    \EndIf
    
    \State \(C \gets \{g \in M \mid v_1(g) \ge v_2(g)\}\) \Comment{Set of all \emph{compressing} goods; $C \subseteq A_1$}
    \State Choose arbitrary \(\hat{g} \in A_2\)
    \State \(A_1' \gets A_1 \cup \{\hat{g}\},\quad A_2' \gets A_2 \setminus \{\hat{g}\}\) \Comment{$v_1(A_1' \setminus \{\hat{g}\}) \ge v_2(A_2')$}
    
    \For{each \(s \in C\)}\label{step:begin_for}
        \If{\(v_1(A_1' \setminus \{s\}) \ge v_2(A_2' \cup \{s\})\)}
            \State \(A_1' \gets A_1' \setminus \{s\},\quad A_2' \gets A_2' \cup \{s\}\) \Comment{Transfer the good \(s\)}
        \Else
            \State \(s^\ast \gets s\)
            \State \textbf{break}
        \EndIf
    \EndFor\label{step:end_for}
    
    \If{\((A_1', A_2')\) is \(\EQ{1}\)}
        \State \Return \((A_1', A_2')\) \label{line:alltransferred}
    \Else\label{line:s-star-start}
        \State \(A_2'' \gets A_1' \setminus \{s^\ast\},\quad A_1'' \gets A_2' \cup \{s^\ast\}\) \Comment{Transfer \(s\) and then swap the bundles}
        \State \Return \((A_1'', A_2'')\) \label{line:s-star-end}
    \EndIf
    \end{algorithmic}
    \end{algorithm}

An \EQ{} + \EQX{} allocation for two agents may however not exist, even with normalised valuations and just three items. Consider the instance shown in \Cref{fig:twoagentseqx}. It can be verified that the only \EQX{} allocation $A$ is when agent 1 gets $A_1 = \{g_3\}$ and agent 2 gets $A_2 = \{g_1, g_2\}$. But then $v_1(A_1) = 5$, $v_2(A_2) = 7$, and there is no distribution over \EQX{} allocations that gives the agents equal utility in expectation.




\begin{figure}[H]
\centering
\begin{tabular}{c|ccc}
      & $g_1$ & $g_2$ & $g_3$ \\ \hline
$1$ & 1     & 3     & 5     \\
$2$ & 4     & 3    & 2    \\ 
\end{tabular}
\caption{An instance demonstrating non-existence of \EQ{} + \EQX{} allocations for two agents.}
\label{fig:twoagentseqx}
\end{figure}

\section{Binary Valuations}
\label{sec:binary}

We now consider instances where agents have binary valuations. As mentioned, it is not difficult to obtain an \EQ{} + \EQ{1} allocation in this case, where goods are assigned to agents that have value 0 for them. However such an allocation is clearly wasteful.

Instead, we show a stronger result. For an instance $\mathcal{I}$, let \textsc{OPT} be the maximum social welfare (i.e., the total value of the agents) in any \EQ{} allocation (possibly fractional). We show that in fact there exists an \EQ{} + \EQ{1} allocation where the social welfare is \textsc{OPT}. Thus, the restriction that the ex ante \EQ{} allocation is supported on \EQ{1} allocations does not impose any cost on the social welfare.

\begin{theorem}
\label{thm:bobw_binary}
Given a fair division instance with binary valuations, an \EQ{} $+$ \EQ{1} allocation $X$ that obtains maximum utilitarian social welfare over all fractional \EQ{} allocations exists, and can be computed efficiently.
\end{theorem}

Note that this result does not require valuations to be normalised.

We first describe the main tool from~\cite{10.1257/aer.103.2.585} we use.
\footnote{Budish et al.~\cite{10.1257/aer.103.2.585} describe Theorem~\ref{thm:decompositionbudish} differently, in terms of assignment matrices and quotas. The description here is adapted to our notation, but the technical content is the same.} Given a binary matrix $G \in \{0,1\}^{n' \times m'}$, for each row $i \in [n']$, define the set $S_i = \{j: A_{ij} = 1\}$ as the columns with non-zero entries. Then a set $\mathcal{T} \subseteq [n']$ is \emph{hierarchical} (or laminar) if for any $i, i' \in \mathcal{T}$, the sets $S_i$ and $S_{i'}$ are either disjoint, or one is contained in the other. The matrix $G$ is \emph{bihierarchical} if the set $[n']$ can be partitioned into $\mathcal{T}_1$ and $\mathcal{T}_2$ so that both $\mathcal{T}_1$ and $\mathcal{T}_2$ are hierarchical.





\begin{theorem}[\cite{10.1257/aer.103.2.585}]
\label{thm:decompositionbudish} Given a binary matrix $G \in \{0,1\}^{n' \times m'}$ and integral vectors $\bar{q}, \ubar{q} \in \mathbb{Z}^{n'}$ such that $G$ is bihierarchical, if the polytope $\{x \in \mathbb{R}^{n'} : \ubar{q} \le Gx \le \bar{q} \}$ is feasible, then it is integral. Further, any fractional solution $x$ can be decomposed into a convex combination of integral solutions in strongly polynomial time. 
\end{theorem}

We are now ready to prove the theorem.

\begin{proof}[Proof of \Cref{thm:bobw_binary}]
    We first write the following LP $L_1$ that maximizes the utilitarian social welfare among all fractional \EQ{} allocations. Note that in an \EQ{} allocation, each agent has the same welfare (captured by the variable $w$), and the social welfare is $n w$. Further, each $v_i(g) \in \{0,1\}$.
    
\begin{align}
    \max \quad & w \\
    \text{subject to:} \quad 
    & \label{constraint:two} \sum_{g=1}^{m} v_{i}(g) x_{ig} = w, \quad \forall i \in N \\
    & \label{constraint:three}  \sum_{i=1}^{n} x_{ig} = 1, \quad \forall g \in M \\
    & \label{constraint:four} x_{ig} \geq 0, \quad \forall i \in N, g \in M
\end{align}

\begin{restatable}{claim}{binaryfeasible}
The LP $L_1$ is feasible.
\label{clm:binaryfeasible}
\end{restatable}

\begin{proof}
    Consider the following assignment. For all $g \in [m]$ such that there is an agent $i$ with $v_i(g)= 0$, assign $x_{ig}= 1$ and $x_{jg} = 0$ for all $j \neq i$. For all the remaining goods, we have that $v_i(g) = 1 ~\forall~ i \in N$. For all such $g$, assign $x_{ig} = \frac{1}{n}$. It is easy to see the corresponding allocation gives a utility of $\frac{1}{n} \cdot k$ to every agent, where $k$ is the number of items valued at $1$ by all the agents. So constraint (\ref{constraint:two}) is satisfied. Constraints (\ref{constraint:three}) and  (\ref{constraint:four}) are satisfied by construction. Thus LP $\mathcal{L}_1$ is feasible. 
\end{proof}

Let $X^*, w^*$ be an optimal solution to $L_1$. Clearly $X^*$ is a fractional \EQ{} allocation of maximum social welfare $w^*$. We will show that $X^*$ can be obtained as a distribution over \EQ{1} allocations, with the same expected welfare. For this, consider the polytope $P_2$, defined as the set of feasible solutions to the following linear constraints.


\begin{align}
    & \label{constraint:six} \lfloor w^* \rfloor \le \sum_{g=1}^m v_i(g) x_{ig} \leq \lceil w^* \rceil , \quad \forall i \in N \\
    \quad 
    & \label{constraint:eight}  \sum_{i=1}^{n} x_{ig} = 1, \quad \forall g \in M \\
    & \label{constraint:nine}  x_{ig} \geq 0, \quad \forall i \in N, g \in M
\end{align}

Comparing $L_1$ and $P_2$, for each agent $i \in N$, $\sum_g v_i(g) X^*_{ig} = w^* \in [\lfloor w^* \rfloor, \lceil w^* \rceil]$. Hence, $X^*$ is a feasible fractional solution to $P_2$. In order to apply Theorem~\ref{thm:decompositionbudish}, we need to show that the constraint matrix is bihierarchical. For this, we define the two hierarchical constraint sets $\mathcal{T}_1$ and $\mathcal{T}_2$ as follows. Let $\mathcal{T}_1$ contain the welfare constraints~(\ref{constraint:six}). Note that each variable $x_{ig}$ appears at most once in the constraints~(\ref{constraint:six}), hence these sets are clearly laminar (and in fact, any two sets are disjoint). Let $\mathcal{T}_2$ contain the assignment constraints~(\ref{constraint:eight}) and nonnegativity constraints~(\ref{constraint:nine}). Then again, in the assignment constraints, each variable appears exactly once, hence these sets are disjoint. The nonnegativity constraints each have size 1, and hence each set is contained by a set in an assignment constraint. Thus, $\mathcal{T}_2$ is laminar as well, and hence the constraint matrix is bihierarchical.

    We can now apply Theorem~\ref{thm:decompositionbudish}, and obtain that the fractional \EQ{} allocation $X^*$ can be decomposed in strongly polynomial time into a convex combination of integral allocations. Let $A^1$, $\ldots$, $A^\ell$ be these integral allocations with convex coefficients $p_1$, $\ldots$, $p_\ell$. Note that each integral allocation $A^{k}$ for $k \in [\ell]$ is an \EQ{1} allocation, since by the welfare constraints~(\ref{constraint:six}), each agent $i$ has value either $\lfloor w^* \rfloor$ or $\lceil w^* \rceil$.

    Finally, note that the expected welfare of the distribution where $A^k$ is drawn with probability $p_k$ is $w^*$. This is because for each $i \in N$ and $g \in M$, $X^*_{ig} = \sum_{k=1}^{\ell} p_k A^{k}_{ig}$, and hence for each agent $i$, the expected value is
    \[
    \sum_{k=1}^\ell p_k \sum_{g \in M} v_i(g) A^{k}_{ig} ~ =~  \sum_{g \in M} v_i(g) \sum_{k=1}^\ell p_k A^k_{ig} ~ = ~ \sum_{g \in M} v_i(g) X^*_{ig} = w^* \, .
    \]
 \end{proof}

\section{General Instances}
\label{sec:hardness}

We now discuss instances with more than $2$ agents and valuations beyond binary. We first show that, beyond binary valuations, the existence of \EQ{} + \EQ{1} allocations is not guaranteed, even in instances with just three agents and four items. The result holds true even if there are only two types of agents and two types of items. The instance is thus tight, in multiple regards: an \EQ{} + \EQ{1} allocation exists with two agents (\Cref{thm:bobw_n_2}), with three agents and three items (\Cref{prop:trivial}), with a single type of agent (i.e., when all agents are identical,~\Cref{prop:trivial}), and with a single type of good (for normalised valuations, this implies that each value $v_i(g)$ is the same).

\begin{theorem}
\label{thm:bobw_n_3}
    Given a fair division instance with $3$ agents, $4$ items, and normalised additive valuations, an \EQ{} $+$ \EQ{1} allocation may not exist. 
\end{theorem}

\begin{proof} 
Consider the instance $\mathcal{I}$ in \Cref{fig:noBOBW} with $3$ agents and $4$ items. 

\begin{figure}[h]
\centering
\begin{tabular}{c|cccc}
      & $g_1$ & $g_2$     & $g_3$ & $g_4$ \\ \hline
$1$ & $1.4$   & $2.2$     & $2.2$ &  $2.2$  \\
$2$ & $5$  & $1$   & $1$ &  $1$     \\
$3$ & $5$  & $1$  & $1$ &  $1$     
\end{tabular}
\caption{An instance with normalised valuations where there is no \EQ{} + \EQ{1} allocation.}
\label{fig:noBOBW}
\end{figure}

To prove the theorem, we will show the claim that in every \EQ{1} allocation $A$, $v_1(A_1) - \frac{1}{2} \left(v_2(A_2) + v_3(A_3)\right) < 0$. It then follows from \Cref{thm:any_convex_comb} that there is no \EQ{} + \EQ{1} allocation for this instance.




To prove the claim, we first show that under any \EQ{1} allocation $A$, agent $1$ gets exactly one good. Since there are $4$ goods and $3$ agents, every agent must get at least $1$ good under $A$. Suppose $a_1$ gets two goods, hence $v_1(A_1) \ge 3.6$. Then agents $2$ and $3$ receive exactly one good. Since both cannot get good $g_1$, assume wlog that $v_2(A_2) = 1$. But then allocation $A$ cannot be \EQ{1}, since on removing any good, agent 1 still has value at least $1.4 > v_2(A_2)$. 

Hence, $|A_1| = 1$. If $A_1 = \{g_1\}$, then $A_2 \cup A_3 = \{g_2, g_3, g_4\}$ and $v_2(A_2) + v_3(A_3) = 3$. Hence $v_1(A_1) = 1.4 < 1.5 = (v_2(A_2) + v_3(A_3))/2$. Similarly, if $A_1 \in \{g_2,g_3, g_4\}$, then $v_1(A_1) = 2.2$, and $v_2(A_2) + v_3(A_3) = 7 > 2 \times v_1(A_1)$. In either case, the claim holds.
\end{proof}







\noindent We next show that even deciding whether an \EQ{} + \EQ{1} allocation exists is NP-hard.

\begin{restatable}{theorem}{bobwnphardweak}
   Given a fair division instance, deciding the existence of an \EQ{} $+$ \EQ{1} allocation is weakly NP-Complete, even for three agents.
   \label{thm:bobw-weak}
\end{restatable}

\begin{proof}
    We exhibit a reduction from \textsc{2-partition}. In this problem, we are given $S = \{b_1, b_2, \ldots b_m\}$ a set of
    $m$ integers with $\sum_{i=1}^m b_i = 2T$. We assume that $m \geq 20$. The task is to decide if there is a partition of the indices $[m]$ into two subsets $S_1$ and $S_2$ such that the sum of the numbers in both partitions equals $T$. That is, $[m] = S_1 \cup S_2$ and $\sum_{i \in S_1} b_i = \sum_{i \in S_2}b_i = T$. Given such an instance of \textsc{2-partition}, we construct a fair division instance as follows. We create $3$ agents, $m$ items $\{g_1, \ldots g_m\}$ (called \emph{partition} items) and two additional items $d_1$ and $d_2$. The valuations are depicted in \Cref{figure:weakhardnessEQEQ1}. Note that the valuations are normalised, and every agent values the grand bundle at $(m+5)T$.

    \begin{figure}[h]
    \centering
    \begin{tabular}{c|ccccccc}
          & $g_1$ & $g_2$     & $\ldots$ & $g_{m-1}$ & $g_m$ & $d_1$ & $d_2$ \\ \hline
    1 & $T$   & $T$     & $\ldots$ &  $T$ & $T$ & $4T$ & $T$ \\
    2 & $b_1$  & $b_2$   & $\ldots$ & $b_{m-1}$ & $b_m$ & $5T$ & $(m-2)T$    \\
    3 & $b_1$  & $b_2$  & $\ldots$ & $b_{m-1}$ & $b_m$ & $5T$ & $(m-2)T$     
    \end{tabular}
    \caption{\centering Reduction for the proof of \Cref{thm:bobw-weak}.}
    \label{figure:weakhardnessEQEQ1}
    \end{figure}
    
    This completes the construction. We now argue the equivalence of the reduction. 
    
    \vspace{0.2cm}
    \noindent \textbf{Forward Direction.} Suppose the given instance is a `yes' instance of 2-partition, and $S_1$ and $S_2$ are the two required partitions. Then, consider the random allocation $X$, defined as follows. In $X$, agent 1 gets $\{d_2\}$, agent 2 gets $\{g_i\}_{i \in S_1}$, and agent 3 gets $\{g_i\}_{i \in S_2}$ with with probability $1$. The remaining item $d_1$ is allocated to agent $1$ with probability $\frac{5}{13}$, and to agents 2 and 3 with probability $\frac{4}{13}$ each.
    Then, $$\mathbb{E}[v_1(X_1)]= 1 \cdot v_1(d_2) +  \frac{5}{13}\cdot v_1(d_1) = 1 \cdot T + \frac{5}{13}\cdot 4T = \frac{33T}{13}$$
    
    $$\mathbb{E}[v_2(X_2)] = 1 \cdot v_2(\{g_i\}_{i \in S_1}) +  \frac{4}{13}\cdot v_2(d_1)  = 1 \cdot T + \frac{4}{13}\cdot 5T = \frac{33T}{13}$$
    
    $$\mathbb{E}[v_3(X_3)] =1 \cdot v_3(\{g_i\}_{i \in S_2}) +  \frac{4}{13}\cdot v_3(d_1) =  1 \cdot T + \frac{4}{13}\cdot 5T = \frac{33T}{13}$$
    
    Therefore, $\mathbb{E}[v_1(X_1)] = \mathbb{E}[v_2(X_2)] = \mathbb{E}[v_3(X_3)]$ and hence, $X$ is an ex-ante \EQ{} allocation.
    
    To see that $X$ is ex post \EQ{1}, note that the support of $X$ has $3$ integral allocations $\{A^1, A^2, A^3\}
    $, where $d_1$ is allocated to agent agent $k$ in the allocation $A^k$ for $k \in [3]$, and $d_2$, $\{g_i\}_{i \in S_1}$ and $\{g_i\}_{i \in S_2}$ are allocated to agents 1, 2, and 3 respectively in all the three allocations. 
    
    
    It is easy to see that all the three integral allocations $\{A
    ^k\}_{k \in [3]}$ are \EQ{1}. Indeed, under any allocation $A^k$, every agent has a utility at least $T$, and for any agent $i$, $v_i(A_i^k \setminus \{d_2\}) = T$.
    
    \vspace{0.2cm}
    \noindent \textbf{Reverse Direction.}
    Suppose the given instance of \textsc{2-partition} is a `no' instance, then we will show that there is no \EQ{} + \EQ{1} allocation for the reduced instance.
    
    To that end, we first show the following claim. 
    
    \begin{claim}
    \label{claim:no_instance_partition}
        If the given instance of \textsc{2-partition} is a `no' instance, then under every \EQ{1} allocation $A$ for the reduction, $v_1(A_1) < \frac{1}{2}v_2(A_2) + \frac{1}{2} v_3(A_3)$.
    \end{claim}
    
    \begin{proof}
    Since the given instance is a `no' instance, under any \EQ{1} allocation $A$, agent $1$ can not receive both $d_1$ and $d_2$.  Otherwise, $v_1(A_1 \setminus d_1) \geq T$ and consequently, the \EQ{1} property requires both agents $2$ and $3$ receive a utility of at least $T$, violating the fact that the given instance is a `no' instance.
    Therefore, agent $1$ can not receive both $d_1$ and $d_2$ in any \EQ{1} allocation. 
    
    We now argue that if $d_2$ is not allocated to agent $1$, then $v_1(A_1) < \frac{1}{2}v_2(A_2) + \frac{1}{2}v_3(A_3)$ and we are done. Indeed, suppose $d_2$ is allocated to agent 2 (the case when $d_2$ is allocated to agent 3 is symmetric).
    Consider the following cases depending on the allocation of $d_1$.
    
    \begin{enumerate}
    \item Suppose $d_1$ is allocated to agent $1$. Then, agent $3$ can derive a utility of at most $2T$ from the set of partition items $\{g_1, \ldots g_m\}$, and hence under any \EQ{1} allocation, agent $1$ can get at most one partition item. Therefore, since $m \geq 20$, we have $v_1(A_1) < 5T < \frac{1}{2}\left(v_2(A_2) + v_3(A_3)\right)$. 
    
    \item Otherwise, if $d_1$ is allocated to agent 3, then since at least partition item must go to agent 1, $v_3(A_3) < 7T$ and hence, agent $1$ cannot receive more than partition $7$ items. Therefore, since $m \geq 20$, we have $v_1(A_1) < 8T < \frac{1}{2}\left( v_2(A_2) + v_3(A_3)\right)$. 
    \end{enumerate}
    
    Therefore, if $d_2$ is allocated to either agent 2 or agent 3, we have $v_1(A_1) < \frac{1}{2}\left( v_2(A_2) + v_3(A_3)\right)$ and we are done. Hence, we can assume that $d_2$ is allocated to agent 1. 
    
    Since both $d_1$ and $d_2$ can not be allocated to agent 1, assume $d_1$ is allocated to agent 2 (the case when $d_1$ is allocated to agent 3 is symmetric).  Note that 
    agent 1 cannot get more than one partition item, otherwise, $v_1(A_1) \geq 3T$ and $v_3(A_3) < 2T$ and the allocation is not \EQ{1}. Therefore, agent $1$ gets at most one partition item. This implies that under any \EQ{1} allocation, $v_1(A_1) \leq 2T < \frac{1}{2}\left( v_2(A_2) + v_3(A_3) \right)$.
    Hence, the claim stands proved.
    \end{proof}
    
    By \Cref{claim:no_instance_partition}, under every \EQ{1} allocation $A$, we have $v_1(A_1) < \frac{1}{2}v_2(A_2) + \frac{1}{2} v_3(A_3)$ if the given instance of \textsc{2-partition} is a `no' instance, and hence in this case it follows from \Cref{thm:any_convex_comb} that there is no \EQ{} + \EQ{1} allocation. This completes the proof of weak NP-hardness.
\end{proof}

\begin{figure}[H]
\centering
\begin{tabular}{c|ccccccc}
      & $g_1$ & $g_2$     & $\ldots$ & $g_{m-1}$ & $g_m$ & $d_1$ & $d_2$ \\ \hline
$a_1$ & $T$   & $T$     & $\ldots$ &  $T$ & $T$ & $4T$ & $T$ \\
$a_2$ & $b_1$  & $b_2$   & $\ldots$ & $b_{m-1}$ & $b_m$ & $5T$ & $(m-2)T$    \\
$a_3$ & $b_1$  & $b_2$  & $\ldots$ & $b_{m-1}$ & $b_m$ & $5T$ & $(m-2)T$     
\end{tabular}
\caption{\centering Reduction for the proof of \Cref{thm:bobw-weak}.}
\label{fig:weakhardnessEQEQ1}
\end{figure}

When the number of agents is not constant, the problem is strongly NP-hard.

\begin{restatable}{theorem}{bobwnphardstrong}
\label{thm:strong_NPC}
    Given a normalised fair division instance, deciding the existence of an \EQ{} $+$ \EQ{1} allocation is strongly \(\NP\)-complete.
\end{restatable}

\begin{proof}
    We present a reduction from the \textsc{3-partition} problem. For this problem, we are given $S = \{b_1, b_2, \ldots, b_m\}$, a set of $m$ integers, where \(m = 3k\) for some \(k\in \mathbb{Z}^+\) . The sum of the integers in \(S\) satisfies \(\sum_{i \in [m]} b_i = kT\). The objective of the \textsc{3-partition} problem is to partition \([m]\) into exactly \(k\) disjoint subsets, \(S=S_1 \uplus S_2 \uplus \cdots \uplus S_k\), such that \(\sum_{i\in S_j} b_i = T\) for all \(j\in [k]\). This problem is known to be strongly \(\NP\)-hard \cite[p.~96--105]{gareyjohnsonComputersintractability09}.\footnote{We consider the unrestricted output variant, where there is no constraint on the cardinality of each subset. This version is also known to be strongly \(\NP\)-hard \cite{gareyjohnsonComputersintractability09}.}
     
     Given such an instance of \textsc{3-partition}, we construct an equivalent fair division instance involving \(k+1\) agents and \(m+2\) goods. The valuations of the agents are specified in Table~\ref{tab:table_reduction_strong2}.
 
 \begin{figure}[h]
     \centering
     \renewcommand{\arraystretch}{1.3} 
     \begin{tabular}{c|ccccccc}
           & $g_1$ & $g_2$     & $\ldots$ & $g_{m-1}$ & $g_m$ & $d_1$ & $d_2$ \\
           \hline
     $0$ & $\frac{2m}{3}T$   & $\frac{2m}{3}T$     & $\ldots$ &  $\frac{2m}{3}T$ & $\frac{2m}{3}T$ & $T$ & $\left(\frac{m}{3}-1\right)T$ \\
     $1$ & $b_1$  & $b_2$   & $\ldots$ & $b_{m-1}$ & $b_m$ & $\frac{m^2}{3}T$ & $\frac{m^2}{3}T$    \\
     $2$ & $b_1$  & $b_2$  & $\ldots$ & $b_{m-1}$ & $b_m$ & $\frac{m^2}{3}T$ & $\frac{m^2}{3}T$    \\
     $\vdots$ & $\vdots$ & $\vdots$ & $\ddots$ & $\vdots$ & $\vdots$ & $\vdots$ & $\vdots$ \\
     $k$ & $b_1$  & $b_2$  & $\ldots$ & $b_{m-1}$ & $b_m$ & $\frac{m^2}{3}T$ & $\frac{m^2}{3}T$ \\
     \end{tabular}
     \caption{\centering Reduction for the proof of \Cref{thm:strong_NPC}.}
     \label{tab:table_reduction_strong2}
 \end{figure}
 
     First, suppose the instance of \textsc{3-partition} is a `yes' instance. We can construct \(k+1\) \EQ{1} allocations, \(A^0,A^1, A^2,\ldots,A^{k}\), as follows. In each allocation, for every \(j\in [k]\), the goods in \(S_j\) are assigned to agent \(j\), while \(d_1\) is allocated to agent \(0\). Additionally, in allocation \(A^i\), \(0\le i\le k\), we allocate good \(d_2\) to agent \(i\). In each allocation \(A^i\), all agents except agent \(i\) receive a utility of \(T\). If \(d_2\) is removed from agent \(i\)'s bundle, then agent \(i\)'s utility is also \(T\). Therefore, these allocations are \EQ{1}.
 
     To get an ex ante EQ allocation, we pick the allocation \(A^0\) with probability \(\frac{3m}{4m-3}\) and each \(A^i\) for $i \in [k]$ with probability \(\frac{3m-9}{4m^2-3m}\). 
     For each agent \(i \in [k]\), the expected utility is:
     \begin{align*}
         \mathbb{E}[v_i(X_i)] &=  T + \frac{3m-9}{4m^2-3m}\cdot\left(\frac{m^2}{3}T\right)\\
         &= T + \frac{m^2-3m}{4m-3}\cdot T\\
         &= T+ \frac{3m}{4m-3}\cdot\left(\frac{m}{3}-1\right)T\\
         &= \mathbb{E}[v_0(X_0)]
     \end{align*}
 
     Therefore, we have an \EQ{} +  \EQ{1} allocation.
 
     We now show that if there exists an \EQ{} + \EQ{1} allocation for the instance in the reduction, then the instance of \textsc{3-partition} is a `yes' instance. Let $X$ be an \EQ{} + \EQ{1} allocation for the instance. Then, from Theorem~\ref{thm:any_convex_comb}, there must exist an \EQ{1} allocation where \(v_0(X_0)\ge \frac{1}{k}\sum_{i=1}^k v_i(X_i)\). We now show that in all such \EQ{1} allocations, where \(v_0(X_0)\ge \frac{1}{k}\sum_{i=1}^k v_i(X_i)\), agent \(0\) must receive both \(d_2\) and \(d_1\).
 
     Consider some such allocation \(A\). Suppose some agent  other than agent \(0\), receives \(d_2\) in \(A\). Then, \(\frac{1}{k}\sum_{i=1}^k v_i(X_i) \ge \frac{1}{k}\cdot\frac{m^2}{3}T = mT\). Now agent \(0\) must receive at least two goods from \(\{g_1,\ldots,g_m\}\) so that \(v_0(X_0)\ge \frac{1}{k}\sum_{i=1}^k v_i(X_i)\) is satisfied. As these two goods are of utility \(\frac{2}{3}T\) each, to satisfy the \EQ{1} condition, each agent \(i \in [k]\) must receive a utility of at least \(\frac{2}{3}T\). Therefore, excluding the agents who receives the goods \(d_1\) and \(d_2\), there at at least \(k-2\) agents who must receive a utility of at least \(\frac{2}{3}T\) from the remaining \(m-2\) goods from \(\{g_1,
     \ldots, g_m\}\). This is not possible, as for any \(m>8\), we have \(\frac{2}{3}mT\cdot {(k-2)} = \frac{2}{3}mT\cdot \left(\frac{m}{3} -2 \right)>\frac{m}{3}T = \sum_{i=1}^m b_i\). 
 
     Similarly, if good \(d_1\) is allocated to some agent other than agent \(0\), then agent \(0\) must get at least two goods from \(\{g_1,\ldots,g_m\}\) to satisfy \(v_0(X_0)\ge \frac{1}{k}\sum_{i=1}^k v_i(X_i)\). Note that \(d_2\) along with one of the goods from \(\{g_1,\ldots,g_m\}\) gives a utility of \((m-1)T < \frac{1}{k}\sum_{i=1}^k v_i(X_i) = mT\). Using the same argument as above, we can show that this is not possible.
     
     Therefore, agent \(0\) must receive both \(d_1\) and \(d_2\) in any \EQ{1} allocation where \(v_0(X_0)\ge \frac{1}{k}\sum_{i=1}^k v_i(X_i)\). This means that each agent \(i \in [k]\) must receive a utility of at least \(T\) from the goods in \(\{g_1, \ldots, g_m\}\). This means that the integers in \(S_i\) can be partitioned into \(k\) subsets, each with sum \(T\). Therefore, the instance of \textsc{3-partition} is a `yes' instance.
 \end{proof}

\subsection*{A Pseudopolynomial Time Algorithm}

Given an instance $\langle N, M, \mathcal{V} \rangle$, let $v_{\max} := \max_{i,g} v_i(g)$ be the maximum value for any good. We now show that if the number of agents is fixed, we can in pseudopolynomial time determine if a BoBW allocation exists, and if so, find it. Note that this shows that the weak NP-hardness shown in~\Cref{thm:bobw-weak} for three agents is in fact tight.

\begin{theorem}
    Given an instance  $\langle N, M, \mathcal{V} \rangle$, we can determine existence of an \EQ{} $+$ \EQ{1} allocation in time $\poly\left((m v_{\max})^n\right)$.
    \label{thm:pseudopoly}
\end{theorem}

As before, let $\eqset := \{A : A \text{ is an \EQ{1} allocation}\}$ be the set of all possible \EQ{1} allocations. Define $\vec{v}(A) := (v_1(A_1), \ldots, v_n(A_n))$ as the vector of agent values, and $\vec{v}(\eqset) := \{\vec{v}(A): A \in \eqset\}$.
The theorem follows easily by observing that for $n$ agents, the number of distinct value profiles $(v_i(A_i))_{i \in [n]}$ over \emph{all} allocations is at most $(m v_{\max})^n$, since each agent's value lies between $0$ and $m v_{\max}$. Then $|\vec{v}(\eqset)| \le (m v_{\max})^n$. Given the set $\vec{v}(\eqset)$, for~\Cref{thm:pseudopoly}, we only need to determine if there exists a convex combination of the vectors in $\vec{v}(\eqset)$ where each entry in the resulting vector is equal. Clearly, such a convex combination exists iff there exists an \EQ{} + \EQ{1} allocation (the convex coefficients give us the probability distribution for the randomized \EQ{} allocation).

We thus need only to obtain the set $\vec{v}(\eqset)$, for which we next give a dynamic program. To simplify the notation, we describe the program for the case $n=2$. This is easily generalised for more agents.

We maintain a table $\mcT \in m \times [m v_{\max}]^2 \times [v_{\max}]^2$ where an entry $\mcT(t,w_1,w_2,h_1,h_2) = 1$ if there exists an allocation $A = (A_1, A_2)$ of the first $t$ items so that $v_1(A_1) = w_1$, $v_2(A_2) = w_2$, $\max_{g \in A_1} v_1(g) = h_1$, and $\max_{g \in A_2} v_2(g) = h_2$. Note that $A$ is \EQ{1} allocation iff $w_i - h_i \le w_{3-i}$ for $i \in \{1,2\}$.

Then for the base case, $\mcT(0,0,0,0,0) = 1$. Inductively, $\mcT(t,w_1,w_2,h_1,h_2) = 1$ if either:

{\setlength{\leftmargini}{5em}
\begin{enumerate}
    \item[(C1a)] \hypertarget{cond:c1a}{} $h_1 = v_1(g_t)$, and for some $h' \le h_1$, $\mcT(t-1,w_1-v_1(g_t),w_2,h',h_2) = 1$, or
    \item[(C1b)] \hypertarget{cond:c1b}{} $h_1 > v_1(g_t)$, and $\mcT(t-1,w_1-v_1(g_t),w_2,h_1,h_2) = 1$, or
    \item[(C2a)] \hypertarget{cond:c2a}{} $h_2 = v_2(g_t)$, and for some $h' \le h_2$, $\mcT(t-1,w_1,w_2-v_2(g_t),h_1,h') = 1$, or
    \item[(C2b)] \hypertarget{cond:c2b}{} $h_2 > v_2(g_t)$, and $\mcT(t-1,w_1,w_2-v_2(g_t),h_1,h_2) = 1$.
\end{enumerate}}

We first show that the table is filled in correctly.

\begin{restatable}{claim}{pseudopolytable}
  The table entry $\mcT(t,w_1,w_2,h_1,h_2) = 1$ iff there exists an allocation $A = (A_1, A_2)$ of the first $t$ items so that $v_i(A_i) = w_i$ and $\max_{g \in A_i} v_i(g) = h_i$ for $i \in \{1,2\}$.
  \label{clm:pseudopolytable}
\end{restatable}

\begin{proof}
    The proof is by induction. Assume the claim holds for all entries $\mcT(t-1, \cdot,\cdot,\cdot,\cdot)$. Let $A$ be an allocation of the first $t$ items, and let $w_i = v_i(A_i)$, and $h_i = \max_{g \in A_i} v_i(g)$ for $i \in \{1,2\}$. We will first show that $\mcT(w_1, w_2, h_1, h_2) = 1$. Suppose $g_t \in A_1$. Let $A_1' = A_1 \setminus \{g_t\}$, and $h' = \max_{g \in A_i'} v_i(g)$. Then $\mcT(t-1,w_1-v_1(g_t),w_2,h',h_2) = 1$. If $h' > v_i(g_t)$, then clearly $h' = h_1$ (since $h'$ continues to be the largest value of an item in $A_1$), and condition \hyperlink{cond:c1b}{(C1b)} is satisfied. If $h' \le v_i(g_t)$, then $h_1 = v_i(g_t)$, and condition \hyperlink{cond:c1b}{(C1b)} is satisfied. In either case, $\mcT(w_1, w_2, h_1, h_2) = 1$. The proof if $g_1 \in A_2$ is analogous.

    Now suppose $\mcT(w_1, w_2, h_1, h_2) = 1$. We need to show there exists such an allocation $A$ of the first $t$ items. Suppose condition \hyperlink{cond:c1a}{(C1a)} is satisfied, and let $A'$ be the allocation for $\mcT(t-1,w_1-v_1(g_t),w_2,h',h_2)$. Then let $A_1 = A_1' \cup \{g_t\}$. Since $h' \le h_1 = v_1(g_t)$, now $g_t$ is a maximum value item in $A_1$, and for the allocation $A$, we get $w_i = v_i(A_i)$ and $h_i = \max{g \in A_i} v_i(g)$ for $i \in \{1,2\}$. If condition \hyperlink{cond:c1b}{(C1b)} is satisfied, let $A'$ be the allocation for $\mcT(t-1,w_1-v_1(g_t),w_2,h_1,h_2)$. Letting $A_1 = A_1' \cup \{g_t\}$, since $h_1 = \max_{g \in A_1'} v_1(g) > v_1(g_t)$, we can see the table entry $\mcT(w_1, w_2, h_1, h_2)$ is verified by the allocation $A$. A similar proof holds if condition \hyperlink{cond:c2a}{(C2a)} or \hyperlink{cond:c2b}{(C2b)} is satisfied.    
\end{proof}

We can now prove~\Cref{thm:pseudopoly}.

\begin{proof}[Proof of~\Cref{thm:pseudopoly}.] Given values $(w_1, \ldots, w_n)$, by~\Cref{clm:pseudopolytable} (suitably extended to $n$ agents), there exists an \EQ{1} allocation $A$ with $v_i(A_i) = w_i$ for $i \in [n]$ iff there exist $h_i$s in $[0,v_{\max}]$ so that $\mcT(m,w_1,\ldots, w_n,h_1,\ldots,h_n)=1$, and for all $i,j \in [n]$, $w_i-h_i \le w_j$. Further, the table $\mcT$ has size $m \times (m v_{\max})^n \times (v_{\max})^n$. Hence, we can in time $\poly\left(m v_{\max})^n\right)$ obtain the set $\vec{v}(\eqset)$ (e.g., by checking each entry of $\mcT$ to see if it satisfies $w_i-h_i \le w_j$ for all $i,j \in [n]$). To decide the existence of an  \EQ{} + \EQ{1} allocation, we only need to determine if there exists a convex combination of the vectors in $\vec{v}(\eqset)$ where each entry in the resulting vector is equal. If such a convex combination exists, then clearly there exists \EQ{} + \EQ{1} allocation, and if not, such an allocation does not exist. This can be determined by, e.g., solving a linear program in $|\vec{v}(\eqset)|$ variables and $n + |\vec{v}(\eqset)|$ constraints, which can be done in time $\poly\left((mv_{\max})^n\right)$.
\end{proof}

\begin{remark}
    The pseudopolynomial time algorithm is easily modified to determine the existence \EQ{} + \EQX{} allocations (by letting $h_i$ be the \emph{minimum} value of a good in $A_i$, in table $\mcT$) and to determine the existence of $i$-biased \EQ{1} allocations (by checking each entry of $\vec{v}(\eqset)$, from the proof of~\Cref{thm:pseudopoly}, to see if agent $i$ has highest value).
\end{remark}

Lastly, we show that if there are more than two agents, then an $i$-biased \EQ{1} allocation may not exist, and it is NP-hard to determine if an $i$-biased \EQ{1} allocation exists. Recall that for two agents, an $i$-biased \EQ{1} allocation always exists for $i \in \{1,2\}$, and an $i$-biased \EQX{} allocation may not exist.

\begin{restatable}{theorem}{biasedhardness}
\label{thm:reduction}
    Given an instance with $3$ agents and normalised valuations, an $i$-biased \EQ{1} allocation may not exist. Deciding the existence of such an allocation is NP-hard.
\end{restatable}
\begin{proof}
    Consider the following instance $\mathcal{I}$ in \Cref{fig:nobiasedEQ1} with $3$ agents and $3$ items. We claim that there is no \EQ{1} allocation $A$ such that $v_1(A_1) \geq v_i(A_i)$ for $i \in \{2, 3\}$. That is, there is no $1$-Biased \EQ{1} allocation.

    \begin{figure}[H]    
    \centering
    \begin{tabular}{c|ccc}
          & $g_1$ & $g_2$ & $g_3$ \\ \hline
    $1$ & 9     & 6     & 6     \\
    $2$ & 1     & 10    & 10    \\
    $3$ & 7     & 7     & 7    
    \end{tabular}
    \caption{An instance with normalised valuations where there is no $1$-Biased \EQ{1} allocation.}
    \label{fig:nobiasedEQ1}
    \end{figure}

    Note that since $m=n$, under any \EQ{1} allocation, every agent must get exactly one item. If agent 1 gets $g_1$ then regardless of what agent $2$ gets, $v_1(A_1) = 9 < 10 = v_2(A_2)$. Alternatively, if agent 1 gets $g_2$ or $g_3$, then $v_1(A_1) = 6 < 7 = v_3(A_3)$. Therefore, there is no \EQ{1} allocation wherein agent 1 is a rich agent.

    To show the hardness of deciding the existence of $i$-biased \EQ{1} allocation, we exhibit a reduction from \textsc{2-partition}, where given a multiset $S = \{b_1, b_2, \ldots b_m\}$ of positive integers with sum $2T$, the task is to decide if there is a partition of $S$ into two subsets $S_1$ and $S_2$ such that the sum of the numbers in both the partitions equals $T$. That is, $S = S_1 \cup S_2$ and $\sum_{b \in S_1} b = \sum_{b' \in S_2}b' = T$. Given an instance of \textsc{2-partition}, we construct a fair division instance as follows. We create $3$ agents, $m$ items $\{g_1, \ldots g_m\}$ that we call \emph{partition items}, and a dummy item $d$. The first agent values all the $m+1$ items at $T$, while the remaining two (identical) agents value each partition item $g_i$ at $s_i$ and the dummy item at $(m-1)T$. The valuations are also depicted in \Cref{fig:table_reduction}. Note that the valuations are normalised, since every agent values the grand bundle at $(m+1)T$.

    \begin{figure}[H]
    \centering
    \begin{tabular}{c|ccccc}
          & $g_1$ & $g_2$     & $\ldots$ & $g_m$ & $d$ \\  \hline
    $1$ & $T$   & $T$     & $\ldots$ &  $T$ & $T$ \\
    $2$ & $b_1$  & $b_2$   & $\ldots$ &  $b_m$ & $(m-1)T$    \\
    $3$ & $b_1$  & $b_2$  & $\ldots$ &  $b_m$ & $(m-1)T$     
    \end{tabular}
    \caption{\centering Reduced instance as in the proof of \Cref{thm:reduction}.}
    \label{fig:table_reduction}
    \end{figure}

    This completes the construction. We now argue the equivalence of the reduction.

    \noindent \textbf{Forward Direction.} Suppose the instance of \textsc{2-partition} is a `yes' instance and say $S_1$ and $S_2$ are the two partitions of $S$ such that $\sum_{b \in S_1} b = \sum_{b' \in S_2}b' = T$. Then consider the allocation $A = \{A_1, A_2, A_3\}$ such that $A_1 = \{d\}$, $A_2 = \{g: v_2(g) \in S_1\}$ and $A_3 = \{g: v_3(g) \in S_2\}$. Then, we have $v_i(A_i) = T ~\forall~ i \in [3]$. Therefore, $A$ is an \EQ{} allocation, and no agent has more utility than that of agent $1$. Hence, the allocation instance is also a `yes' instance.

    \vspace{0.2cm}

    \noindent \textbf{Reverse Direction.} Suppose the reduced instance admits a $1$-biased \EQ{1} allocation $A$ such that $v_1(A_1) \geq v_i(A_i)$ for $i \in \{2, 3\}$. Then, $d$ cannot be assigned to either agent 2 or 3. Thus $d \in A_1$ and $v_1(A_1) \geq T.$ If agent 1  also gets any of the partition items, say $g$, in addition to $d$, then $v_1(A_1) \geq 2T$. Since $A$ is \EQ{1}, it must be the case that $v_2(A_2) \geq T$ and $v_3(A_3) \geq T$. Then, $v_2(A_2) + v_3(A_3) = v_2(A_2 \cup A_3) = v_2(S \setminus g) \geq 2T$. But this is a contradiction since $v_2(S \setminus g) < v_2(S) = 2T$. Therefore, we have $A_1 = \{d\}$. Also, since $v_1(A_1) \geq v_i(A_i)$ for $i \in \{2, 3\}$, the utility of agents 2 and 3 is at most $T$. Since $A$ is a complete allocation, we have $v_2(A_2) = v_3(A_3) = T$. This corresponds to a partition of $S$ into $S_1 = \{v_2(g): g \in A_2\}$ and $S_2 = \{v_3(g): g \in A_3\}$. Hence, we have $\sum_{b \in S_1} b = \sum_{b' \in S_2}b' = T$, and hence the partition instance is also a `yes' instance.

    The theorem stands proved.
\end{proof}



\section{Conclusion}

Our work gives a geometric characterization for the BoBW equitability, and presents a complete landscape of the existence and computation of such allocations, both for \EQ{} + \EQ{1} and \EQ{} + \EQX{}. We also study $i$-biased \EQ{1} allocations, which can be of independent interest for other fairness notions, as well as an independent measure of fairness. An obvious open question regards BoBW allocations with approximate equitability. Concretely, can we show existence of ex ante $\alpha$-EQ and ex post \EQ{1} (or \EQX{}) allocations, for $\alpha < 1$? An ex ante $\alpha$-EQ allocation would require that in expectation, $\alpha v_i(X_i) \le v_j(X_j) \le \frac{1}{\alpha} v_i(X_i)$ for all agents $i$, $j$. It would also be interesting to study if our results for binary valuations extend beyond additive valuations, e.g., for matroid rank valuations. Finally, given the nonexistence of BoBW allocations for general instances, it would be useful to study other restricted domains where such allocations do exist.

\subsection*{Acknowledgments}

This research was supported in part by the Department of Atomic Energy, Government of India, under Project No. RTI4001. Vishwa Prakash HV acknowledges the support from TCS Research Fellowship. Aditi Sethia acknowledges the support from the Walmart Center for Tech Excellence (CSR WMGT-23-
0001) 


\bibliographystyle{alphaurl}
\bibliography{References}
\end{document}